\documentclass[journal]{IEEEtran} 

\usepackage{bbm} % For sets and Expectation. Double-striked left side and normal right side fonts

\usepackage{bm} % For Vectors, Matrices. Command \bm

\usepackage{amsmath} % American Math Society Math

\usepackage{amsthm} % Theorem Like Definition

\usepackage{amssymb} % Extended symbol collection

\usepackage{mathrsfs} % For the command \mathscr - Ralph Smith's Formal Script font

\usepackage{optidef}  % Standard set of environments for writing optimization problems.

\usepackage{algorithm} % Algorithm

\usepackage[noend]{algpseudocode} % Pseudocode

\usepackage{cite} % Cite

\usepackage{graphicx} % Inserting images

\usepackage{textcomp} % Special Symbols

\usepackage{xcolor} % Colourful text

\usepackage{caption} % Caption
\usepackage{subcaption} % Sub Captions

\usepackage{hyperref} % Clickable References

\newtheorem{theorem}{Theorem}
\newtheorem{lemma}{Lemma}

\newtheorem{corollary}{Corollary}
\newtheorem{definition}{Definition}

\newtheorem{remark}{Remark}

\newcommand\ddfrac[2]{\frac{\displaystyle #1}{\displaystyle #2}}

\begin{document}
\title{Analysis of Age of Incorrect Information under Generic Transmission Delay} 

\author{%
\IEEEauthorblockN{Yutao Chen and Anthony Ephremides}\\
    \IEEEauthorblockA{Department of Electrical and Computer Engineering, University of Maryland}
}

\maketitle

\begin{abstract}
This paper investigates the Age of Incorrect Information (AoII) in a communication system whose channel suffers a random delay. We consider a slotted-time system where a transmitter observes a dynamic source and decides when to send updates to a remote receiver through the communication channel. The threshold policy, under which the transmitter initiates transmission only when the AoII exceeds the threshold, governs the transmitter's decision. In this paper, we analyze and calculate the performance of the threshold policy in terms of the achieved AoII. Using the Markov chain to characterize the system evolution, the expected AoII can be obtained precisely by solving a system of linear equations whose size is finite and depends on the threshold. We also give closed-form expressions of the expected AoII under two particular thresholds. Finally, calculation results show that there are better strategies than the transmitter constantly transmitting new updates.
\end{abstract}

\section{Introduction}
With the development of wireless technology and cheap sensors, real-time monitoring systems are widely used. Typically in such systems, a monitor monitors one or more events simultaneously and transmits updates to allow one or more receivers at a distance to have a good knowledge of the events. Therefore, the timeliness of information is often one of the most important performance indicators. The Age of Information (AoI), first introduced in~\cite{b1}, captures the freshness of information by tracking the time elapsed since the generation of the last received update. More precisely, let $V(t)$ be the generation time of the last update received up to time $t$. Then, AoI at time $t$ is defined by $\Delta_{AoI}(t) = t-V(t)$. After the introduction, it has attracted extensive attention and research~\cite{aoi1,aoi2,aoi3}. However, the limitation of AoI is that it ignores the information content of the transmitted updates. Therefore, it falls short in the context of remote estimation. For example, we want to estimate a rapidly changing event remotely. In this case, a small AoI does not necessarily mean that the receiver has accurate information about the event. Likewise, when the event changes slowly, the receiver can estimate relatively accurately without timely information.

Inspired by the above limitation, the Age of Incorrect Information (AoII) is introduced in~\cite{3}, which combines the timeliness and the information content. As presented in~\cite{3}, AoII is dominated by two penalty functions. The first is the time penalty function, which captures the time elapsed since the last time the receiver has perfect information about the remote event. The second is the information penalty function, which captures the information mismatch between the receiver and the remote event. Therefore, AoII captures not only the information mismatch between the event and the receiver but also the aging process of conflicting information. Moreover, by choosing different penalty functions, AoII is adaptable to various systems and communication goals. Hence, AoII is regarded as a semantic metric~\cite{semantic}.

Several works have been done since the introduction. Optimizing AoII under resource constraints is studied in~\cite{3, b2, b3}, with different assumptions on the system and different penalty function choices. AoII in scheduling problems is studied in~\cite{b4,b5}. However, all these papers assume that the transmission time of each update is deterministic and normalized. In this paper, we consider the communication system in which the transmission time of an update is random. A similar system setup is considered in~\cite{b6}, where the problem is studied based on simulation results. However, we provide a theoretical analysis of the system, and the results apply to generic transmission delay. The system with random transmission delay has also been studied in the context of remote estimation and AoI~\cite{b7,b10,b8,b9}. However, the problem considered in this paper is very different, as AoII is a combination of age-based metrics frameworks and error-based metrics frameworks.

The main contributions of this paper are: 1) We investigate the AoII in the system where the communication channel suffers a random delay. 2) We analyze the characteristics of the threshold policy, under which the transmitter initiates transmission only when AoII exceeds the threshold. 3) We calculate the expected AoII achieved by the threshold policy precisely.

The remainder of this paper is organized in the following way. In Section \ref{sec-system}, we introduce the system model. Then, in Section \ref{sec-MDP}, we theoretically analyze and calculate the expected AoII achieved by the threshold policy. Finally, we conclude the paper with numerical results, which are detailed in Section \ref{sec-Numerical}.

\section{System Overview}\label{sec-system}
\subsection{System Model}\label{sec-SystemModel}
We consider a transmitter-receiver pair in a slotted-time system, where the transmitter observes a dynamic source and needs to send status updates to the remote receiver through a communication channel. The dynamic source is modeled by a two-state symmetric Markov chain with state transition probability $p$. An illustration of the Markov chain is shown in Fig.~\ref{fig-MarkovChain}.
\begin{figure}
\centering \includegraphics[width=3in]{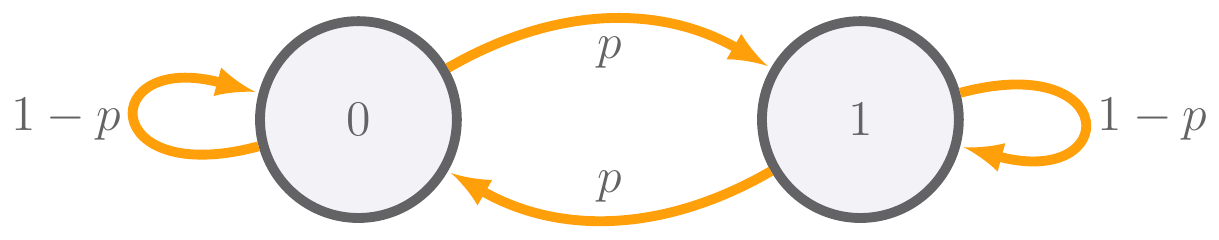} \caption{Two-state symmetric Markov chain with state transition probability $p$.}
\label{fig-MarkovChain}
\end{figure}
The transmitter receives an update from the dynamic source at the beginning of each time slot. We denote the update at time slot $k$ as $X_k$. The old update will be discarded upon the arrival of the new one. Then, the transmitter must decide whether to transmit the update based only on the system's current status. More precisely, when the channel is idle (i.e., no transmission in progress), the transmitter chooses between transmitting the update and staying idle. When the channel is busy, the transmitter cannot do anything other than stay idle. The transmission time $T\in\mathbb{N}^*$ of an update is a random variable with distribution denoted by $p_t\triangleq Pr(T=t)$. We assume that $T$ is independent and identically distributed. The channel is error-free, meaning the update will not be corrupted during transmission. When a transmission finishes, the communication channel is immediately available for the subsequent transmission.

The remote receiver will estimate the state of the dynamic source using the received updates. We denote by $\hat{X}_k$ the receiver's estimate at time slot $k$. According to~\cite{b9}, the best estimator when $p\leq\frac{1}{2}$ is the last received update. When $p>\frac{1}{2}$, the optimal estimator depends on the transmission time. In this paper, we consider only the case of $p\leq\frac{1}{2}$. Hence, the receiver uses the last received update as the estimate. For the case of $p>\frac{1}{2}$, the results can be extended using the corresponding best estimator. The receiver uses $ACK/NACK$ packets to inform the transmitter of its reception of the new update. $ACK/NACK$ packets are generally very small~\cite{2a2}. Hence, we assume they are reliably and instantaneously received by the transmitter. Then, when $ACK$ is received, the transmitter knows that the receiver's estimate has changed to the last sent update. When $NACK$ is received, the transmitter knows that the receiver's estimate did not change. In this way, the transmitter always knows the current estimate on the receiver side. An illustration of the system model is shown in Fig.~\ref{fig-SystemModel}.
\begin{figure}
\centering \includegraphics[width=3in]{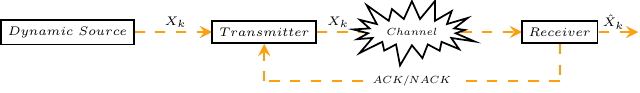}\caption{An illustration of the system model.}
\label{fig-SystemModel}
\end{figure}

\subsection{Age of Incorrect Information}
The system adopts the Age of Incorrect Information (AoII) as the performance metric. We first define $U_k$ as the last time instant before time $k$ (including $k$) that the receiver's estimate is correct. Mathematically,
\[
U_k \triangleq \max\{h:h\leq k, X_h = \hat{X}_h\}.
\]
Then, in a slotted-time system, AoII at time slot $k$ can be written as
\begin{equation}\label{eq-AoII}
\Delta_{AoII}(X_k,\hat{X}_k,k) = \sum_{h=U_k+1}^k\bigg(g(X_h,\hat{X}_h) F(h-U_k)\bigg),
\end{equation}
where $g(X_k,\hat{X}_k)$ is the information penalty function. $F(k) \triangleq f(k) - f(k-1)$ where $f(k)$ is the time penalty function. In this paper, we choose $g(X_k,\hat{X}_k) = |X_k-\hat{X}_k|$ and $f(k) = k$. Hence, $F(k) =1 $ for all $k$ and $g(X_k,\hat{X}_k)\in\{0,1\}$. Then, equation \eqref{eq-AoII} can be simplified as
\[
\Delta_{AoII}(X_k,\hat{X}_k,k) = k-U_k\triangleq\Delta_k.
\]
To characterize the evolution of $\Delta_k$, it is sufficient to characterize the value of $\Delta_{t+1}$ using $\Delta_k$ and the system dynamics. To this end, we distinguish between the following two cases.
\begin{itemize}
\item When the receiver's estimate is correct at time $k + 1$, we have $U_{k+1} = k + 1$. Then, by definition, $\Delta_{k+1} = 0$.
\item When the receiver's estimate is incorrect at time $k + 1$, we have $U_{k+1} = U_k$. Then, by definition, $\Delta_{k+1}=k+1-U_k=\Delta_k +1$.
\end{itemize}
Combining together, we have
\begin{equation}\label{eq-AoIIDynamic}
\Delta_{k+1} = \mathbbm{1}\{U_{k+1}\neq k+1\}(\Delta_k+1),
\end{equation}
where $\mathbbm{1}\{A\}$ is the indicator function, whose value is one when event $A$ occurs and zero otherwise. A sample path of $\Delta_k$ is shown in Fig.~\ref{fig-SamplePath}.
\begin{figure}
\centering \includegraphics[width=3in]{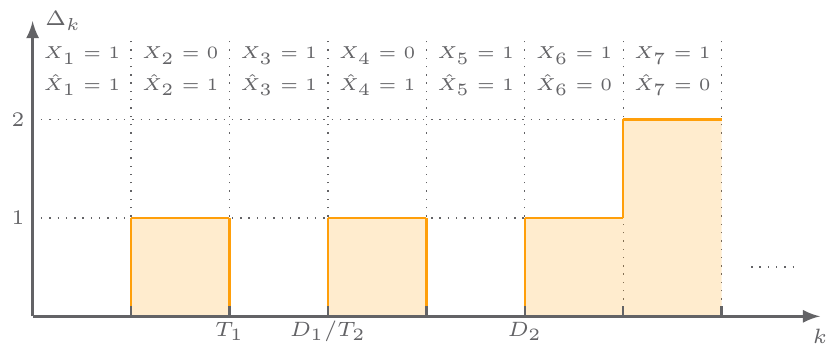}
\caption{A sample path of $\Delta_k$, where $T_i$ and $D_i$ are the transmission start time and the delivery time of the $i$-th update, respectively. At $T_1$, the transmitted update is $X_3$. The estimate at time slot $6$ (i.e., $\hat{X}_6$) changes due to the reception of the update transmitted at $T_2$.}
\label{fig-SamplePath}
\end{figure}
When combined with the source dynamics and the receiver's estimate, the evolution of $\Delta_{k}$ can also be characterized by the following cases.
\begin{itemize}
\item When $\Delta_k=0$ and $\hat{X}_k = \hat{X}_{k+1}$, we know $\hat{X}_k=X_k$ and the receiver's estimate remains the same. Then, when the source remains in the same state (i.e., $X_k = X_{k+1}$), we have $X_{k+1} = \hat{X}_{k+1}$. Hence, $\Delta_{k+1}=0$ with probability $1-p$. On the other hand, when the source changes state (i.e., $X_k \neq X_{k+1}$), we have $X_{k+1} \neq \hat{X}_{k+1}$. Hence, $\Delta_{k+1}=1$ with probability $p$. Combining together, we have
\[
\Delta_{k+1} = \begin{dcases}
0 & w.p.\ 1-p,\\
1 & w.p.\ p.
\end{dcases}
\]
The same analysis can be applied to other cases. Hence, the details are omitted for the following cases.
\item When $\Delta_k=0$ and $\hat{X}_k \neq \hat{X}_{k+1}$, we have
\[
\Delta_{k+1} = \begin{dcases}
0 & w.p.\ p,\\
1 & w.p.\ 1-p.
\end{dcases}
\]
\item When $\Delta_k>0$ and $\hat{X}_k = \hat{X}_{k+1}$, we have
\[
\Delta_{k+1} = \begin{dcases}
0 & w.p.\ p,\\
\Delta_k+1 & w.p.\ 1-p.
\end{dcases}
\]
\item When $\Delta_k>0$ and the $\hat{X}_k \neq \hat{X}_{k+1}$, we have
\[
\Delta_{k+1} = \begin{dcases}
0 & w.p.\ 1-p,\\
\Delta_k+1 & w.p.\ p.
\end{dcases}
\]
\end{itemize}
Combining together, we obtain \eqref{eq-AoIIWithSource}.
\begin{figure*}[!t]
\normalsize
\begin{equation}\label{eq-AoIIWithSource}
\Delta_{k+1} = \begin{dcases}
\mathbbm{1}\{\hat{X}_k\neq\hat{X}_{k+1}; \Delta_k=0\} + \mathbbm{1}\{\hat{X}_k=\hat{X}_{k+1};\Delta_k>0\}(\Delta_k+1) & w.p.\ 1-p,\\
\mathbbm{1}\{\hat{X}_k=\hat{X}_{k+1}; \Delta_k=0\} + \mathbbm{1}\{\hat{X}_k\neq\hat{X}_{k+1};\Delta_k>0\}(\Delta_k+1) & w.p.\ p.\\
\end{dcases}
\end{equation}
\hrulefill
\vspace*{4pt}
\end{figure*}
In this paper, we use the expected AoII to quantify the performance. To this end, we first define a policy $\phi$ as the one that specifies the transmitter's action in each time slot. Then, the expected AoII achieved by policy $\phi$ is
\[
\bar{\Delta}_{\phi} \triangleq\lim_{K\to\infty} \frac{1}{K}\mathbb{E}_{\phi}\left(\sum_{k=0}^{K-1}\Delta_k\right).
\]
To evaluate $\bar{\Delta}_{\phi}$, we start with characterizing the system dynamics in the following subsection.

\subsection{System Dynamics}\label{sec-SystemDynamic}
The first thing to notice is that the transmission time for each update is unbounded. To simplify the analysis, we consider the following two more practical cases\footnote{All results presented in this paper apply to both cases unless stated otherwise.}.
\begin{itemize}
\item \textbf{Assumption 1}: We assume that the transmission will take at most $t_{max}$ time slots and the update will always be delivered. More precisely, we assume $1\leq T\leq t_{max}$ and
\[
\sum_{t=1}^{t_{max}}p_t = 1,\quad p_t\geq0,\ 1\leq t\leq t_{max}.
\]
In practice, we can always choose a sufficiently large $t_{max}$ so that the probability that the transmission time is longer than $t_{max}$ is negligible.
\item \textbf{Assumption 2}: We assume that the transmission survives a maximum of $t_{max}$ time slots. When the transmission lasts to the $t_{max}$th time slot, the update is either delivered or discarded at the end of the $t_{max}$th time slot. In both cases, the channel will be available for a new transmission at the next time slot. In practice, similar techniques, such as time-to-live (TTL)~\cite{b12}, are used to prevent an update from occupying the channel for too long.
\end{itemize}
\begin{remark}
When $t_{max}=1$, the system reduces to the one considered in~\cite{3}. Hence, for the remainder of this paper, we consider the case of $t_{max}>1$.
\end{remark}
We notice that the status of the system at the beginning of the $k$th time slot can be fully captured by the triplet $s_k\triangleq(\Delta_k,t_k,i_k)$ where $t_k\in\{0,1,...,t_{max}-1\}$ indicates the time the current transmission has been in progress. We define $t_k=0$ if there is no transmission in progress. The last element $i_k\in\{-1,0,1\}$ indicates the state of the channel. We define $i_k=-1$ when the channel is idle. $i_k=0$ if the channel is busy and the transmitting update is the same as the receiver's current estimate, and $i_k=1$ when the transmitting update is different from the receiver's current estimate.
\begin{remark}\label{rem-ti}
Note that, according to the definition of $t_k$ and $i_k$, $i_k=-1$ if and only if $t_k=0$. In this case, the channel is idle.
\end{remark}
Then, characterizing the system dynamics is equivalent to characterizing the value of $s_{k+1}$ using $s_k$ and the transmitter's action. We denote by $a_k\in\{0,1\}$ the transmitter's decision. We define $a_k=0$ when the transmitter decides not to initiate a transmission and $a_k=1$ otherwise. Hence, the system dynamics can be fully characterized by $P_{s_k,s_{k+1}}(a_k)$, which is defined as the probability that action $a_k$ at $s_k$ leads to $s_{k+1}$.

\section{AoII Analysis}\label{sec-MDP}
As is proved in~\cite{b2,b3,b4,b5}, the AoII-optimal policy often has a threshold structure. Hence, we consider threshold policy.
\begin{definition}[Threshold policy]\label{def-ThreholdPolicy}
Under threshold policy $\tau$, the transmitter will initiate a transmission only when the current AoII is no less than threshold $\tau\in\mathbbm{N}^0$ and the channel is idle.
\end{definition}
\begin{remark}
We define $\tau\triangleq\infty$ as the policy under which the transmitter never initiates any transmissions.
\end{remark}
We notice that the system dynamics under threshold policy can be characterized by a discrete-time Markov chain (DTMC). Without loss of generality, we assume the DTMC starts at state $s=(0,0,-1)$. Then, the state space $\mathcal{S}$ consists of all the states accessible from state $s=(0,0,-1)$. Since state $s=(0,0,-1)$ is positive recurrent and communicates with each state $s\in\mathcal{S}$, the stationary distribution exists. Let $\pi_{s}$ be the steady-state probability of state $s$. Then, $\pi_s$ satisfies the following balance equation.
\[
\pi_s = \sum_{s'\in\mathcal{S}}P_{s',s}(a)\pi_{s'},\quad s\in\mathcal{S},
\]
where $P_{s',s}(a)$ is the single-step state transition probability, and the action $a$ depends on the threshold policy. Then, the first step in calculating the expected AoII achieved by the threshold policy is to calculate the stationary distribution of the induced DTMC. However, the problem arises as the state space $\mathcal{S}$ is infinite and intertwined. To simplify the state transitions, we recall that the transmitter can only stay idle (i.e., $a=0$) when the channel is busy. Let $\mathcal{S}'=\{s=(\Delta,t,i):i\neq-1\}$ be the set of the state where the channel is busy. Then, for $s'\in\mathcal{S}'$, $P_{s',s}(a) = P_{s',s}(0)$ and is independent of the threshold policy. Hence, for any threshold policy and each $s\in\mathcal{S}\setminus\mathcal{S}'$, we can repeatedly replace $\pi_{s'}$, where $s'\in\mathcal{S}'$, with the corresponding balance equation until we get the following equation.
\begin{equation}\label{eq-CompactBalanceEq}
\pi_{s} = \sum_{s'\in \mathcal{S}\setminus\mathcal{S}'}P_{\Delta',\Delta}(a)\pi_{s'},\quad s\in\mathcal{S}\setminus\mathcal{S}',
\end{equation}
where $P_{\Delta',\Delta}(a)$ is the multi-step state transition probability from state $s'=(\Delta',0,-1)$ to state $s=(\Delta,0,-1)$ under action $a$. For simplicity, we write \eqref{eq-CompactBalanceEq} as
\begin{equation}\label{eq-CompactBalanceEq2}
\pi_{\Delta} = \sum_{\Delta'\geq0}P_{\Delta',\Delta}(a)\pi_{\Delta'},\quad \Delta\geq0.
\end{equation}
As we will see in the following subsections, $\pi_\Delta$ is sufficient to calculate the expected AoII obtained by any threshold policy.
\begin{remark}
The intuition behind the simplification of the balance equations is as follows. We recall that the system dynamics when the channel is busy are independent of the adopted policy. Hence, we can calculate these dynamics in advance so that the balance equations contain only the states in which the transmitter needs to make decisions. 
\end{remark}
\noindent In the next subsection, we derive the expression of $P_{\Delta,\Delta'}(a)$.

\subsection{Multi-step State Transition Probability}\label{sec-StateTransProb}
We start with the case of $a=0$. In this case, no update will be transmitted, and $P_{\Delta,\Delta'}(0)$ is independent of the transmission delay. Then, according to \eqref{eq-AoIIWithSource}, we have
\[
P_{0,\Delta'}(0) = \begin{dcases}
1-p & \Delta'=0,\\
p & \Delta'=1,
\end{dcases}
\]
and for $\Delta>0$,
\[
P_{\Delta,\Delta'}(0) = \begin{dcases}
p & \Delta'=0,\\ 
1-p & \Delta' = \Delta+1.
\end{dcases}
\]
In the sequel, we focus on the case of $a=1$. We define $P^{t}_{\Delta,\Delta'}(a)$ as the probability that action $a$ at state $s=(\Delta,0,-1)$ will lead to state $s'=(\Delta',0,-1)$, given that the transmission takes $t$ time slots. Then, under \textbf{Assumption 1},
\[
P_{\Delta,\Delta'}(1) = \sum_{t=1}^{t_{max}}p_tP^t_{\Delta,\Delta'}(1).
\]
Hence, it is sufficient to obtain the expressions of $P^t_{\Delta,\Delta'}(1)$. To this end, we define $p^{(t)}$ as the probability that the dynamic source will remain in the same state after $t$ time slots. Since the Markov chain is symmetric, $p^{(t)}$ is independent of the state and can be calculated by
\[
p^{(t)} = \left(\begin{bmatrix}
1-p & p\\
p & 1-p
\end{bmatrix}^t\right)_{11},
\]
where the subscript indicates the row number and the column number of the target probability. For the consistency of notation, we define $p^{(0)}\triangleq1$. Then, we have the following lemma.
\begin{lemma}\label{lem-CompactTrans}
Under \textbf{Assumption 1}, 
\begin{equation}\label{eq-Assumption1Trans}
P_{\Delta,\Delta'}(1) = \sum_{t=1}^{t_{max}}p_tP^t_{\Delta,\Delta'}(1),
\end{equation}
where
\[
P^{t}_{0,\Delta'}(1) = 
\begin{dcases}
p^{(t)} & \Delta'=0,\\
p^{(t-k)}p(1-p)^{k-1} & 1\leq\Delta'= k\leq t,\\
0 & otherwise,
\end{dcases}
\]
and for $\Delta>0$,
\begin{multline*}
P^{t}_{\Delta,\Delta'}(1)=\\
\begin{dcases}
p^{(t)} & \Delta'=0,\\
(1-p^{(t-1)})(1-p) & \Delta'=1,\\
(1-p^{(t-k)})p^{2}(1-p)^{k-2} & 2\leq \Delta'=k\leq t-1,\\
p(1-p)^{t-1} & \Delta'=\Delta+t,\\
0 & otherwise.
\end{dcases}
\end{multline*}
\end{lemma}
\begin{proof}
The expression of $P^t_{\Delta,\Delta'}(1)$ is obtained by analyzing the system dynamics. The complete proof can be found in Appendix \ref{pf-CompactTrans} of the supplementary material.
\end{proof}
To get more insights, we provide the following corollary.
\begin{corollary}\label{lem-StateTransProb}
Under \textbf{Assumption 1}, equation \eqref{eq-Assumption1Trans} can be written equivalently as \eqref{eq-EquivalentEq1}
\begin{figure*}[!t]
\normalsize
\begin{equation}\label{eq-EquivalentEq1}
P_{\Delta,\Delta'}(1) =\\
\begin{dcases}
\sum_{t=\Delta'}^{t_{max}}p_tP^t_{\Delta,\Delta'}(1) & 0\leq\Delta'\leq t_{max}-1,\Delta\geq\Delta',\\
\sum_{t=\Delta'}^{t_{max}}p_tP^t_{\Delta,\Delta'}(1) + p_{t'}P^{t'}_{\Delta,\Delta'}(1) & 0\leq\Delta'\leq t_{max}-1,\Delta<\Delta',\\
p_{t'}P^{t'}_{\Delta,\Delta'}(1) & \Delta'\geq t_{max}.
\end{dcases}
\end{equation}
\hrulefill
\vspace*{4pt}
\end{figure*}
where $t'\triangleq\Delta'-\Delta$ and $P^{t'}_{\Delta,\Delta'}(1)\triangleq 0$ when $t'\leq0$ or when $t'>t_{max}$. Meanwhile, $P_{\Delta,\Delta'}(1)$ possesses the following properties.
\begin{enumerate}
\item $P_{\Delta,\Delta'}(1)$ is independent of $\Delta$ when $0\leq\Delta'\leq t_{max}-1$ and $\Delta\geq\Delta'$.
\item $P_{\Delta,\Delta'}(1) = P_{\Delta+\delta,\Delta'+\delta}(1)$ when $\Delta'\geq t_{max}$ and $\Delta\geq0$ for any $\delta\geq1$.
\item $P_{\Delta,\Delta'}(1)=0$ when $\Delta'>\Delta+t_{max}$ or when $t_{max}-1<\Delta'<\Delta+1$.
\end{enumerate}
\end{corollary}
\begin{proof}
The equivalent expression and the properties can be obtained by analyzing the equations detailed in Lemma \ref{lem-CompactTrans}. The complete proof can be found in Appendix \ref{pf-StateTransProb} of the supplementary material.
\end{proof}
The state transition probabilities under \textbf{Assumption 2} can be obtained similarly. To this end, we first define $p_{t^+}\triangleq\sum_{t=t_{max}+1}^{\infty}p_t$ as the probability that the transmission will be terminated and $P^{t^+}_{\Delta,\Delta'}(a)$ as the probability that action $a$ at state $s=(\Delta,0,-1)$ will result in state $s'=(\Delta',0,-1)$ when the transmission is terminated. Then, we have the following lemma.
\begin{lemma}\label{lem-Case2TransProb}
Under \textbf{Assumption 2},
\begin{equation}\label{eq-compactTransProbs}
P_{\Delta,\Delta'}(1) = \sum_{t=1}^{t_{max}}p_tP^{t}_{\Delta,\Delta'}(1) + p_{t^+}P^{t^+}_{\Delta,\Delta'}(1),
\end{equation}
where
\[
P^{t}_{0,\Delta'}(1) = 
\begin{dcases}
p^{(t)} & \Delta'=0,\\
p^{(t-k)}p(1-p)^{k-1} & 1\leq\Delta'= k\leq t,\\
0 & otherwise,
\end{dcases}
\]
\[
P^{t^+}_{0,\Delta'}(1) = P^{t_{max}}_{0,\Delta'}(1),
\]
and for $\Delta>0$,
\begin{multline*}
P^{t}_{\Delta,\Delta'}(1) = \\
\begin{dcases}
p^{(t)} & \Delta'=0,\\
(1-p^{(t-1)})(1-p) & \Delta'=1,\\
(1-p^{(t-k)})p^{2}(1-p)^{k-2} & 2\leq\Delta'=k\leq t-1,\\
p(1-p)^{t-1} & \Delta'=\Delta+t,\\
0 & otherwise,
\end{dcases}
\end{multline*}
\begin{multline*}
P^{t^+}_{\Delta,\Delta'}(1) = \\
\begin{dcases}
1-p^{(t_{max})} & \Delta'=0,\\
(1-p^{(t_{max}-k)})p(1-p)^{k-1} & 1\leq\Delta'= k\leq t_{max}-1,\\
(1-p)^{t_{max}} & \Delta' = \Delta+t_{max},\\
0 & otherwise.
\end{dcases}
\end{multline*}
Under \textbf{Assumption 2}, equation \eqref{eq-compactTransProbs} can be written equivalently as \eqref{eq-EquivalentEq2}.
\begin{figure*}[!t]
\normalsize
\begin{equation}\label{eq-EquivalentEq2}
P_{\Delta,\Delta'}(1) =
\begin{dcases}
\sum_{t=\Delta'}^{t_{max}}p_tP^t_{\Delta,\Delta'}(1)+ p_{t^+}P^{t^+}_{\Delta,\Delta'}(1) & 0\leq\Delta'\leq t_{max}-1,\Delta\geq\Delta',\\
\sum_{t=\Delta'}^{t_{max}}p_tP^t_{\Delta,\Delta'}(1) + p_{t'}P^{t'}_{\Delta,\Delta'}(1)+ p_{t^+}P^{t^+}_{\Delta,\Delta'}(1) & 0\leq\Delta'\leq t_{max}-1,\Delta<\Delta',\\
p_{t'}P^{t'}_{\Delta,\Delta'}(1)+ p_{t^+}P^{t^+}_{\Delta,\Delta'}(1) & \Delta'\geq t_{max},\\
0 & otherwise.
\end{dcases}
\end{equation}
\hrulefill
\vspace*{4pt}
\end{figure*}
Meanwhile, $P_{\Delta,\Delta'}(1)$ possesses the following properties.
\begin{enumerate}
\item $P_{\Delta,\Delta'}(1)$ is independent of $\Delta$ when $0\leq\Delta'\leq t_{max}-1$ and $\Delta\geq\max\{1,\Delta'\}$.
\item $P_{\Delta,\Delta'}(1) = P_{\Delta+\delta,\Delta'+\delta}(1)$ when $\Delta'\geq t_{max}$ and $\Delta>0$ for any $\delta\geq1$.
\item $P_{\Delta,\Delta'}(1)=0$ when $\Delta'>\Delta+t_{max}$ or when $t_{max}-1<\Delta'<\Delta+1$.
\end{enumerate}
\end{lemma}
\begin{proof}
The proof follows similar steps as presented in the proofs of Lemma \ref{lem-CompactTrans} and Corollary \ref{lem-StateTransProb}. The complete proof can be found in Appendix \ref{pf-Case2TransProb} of the supplementary material.
\end{proof}
As the expressions of $P_{\Delta,\Delta'}(a)$ under both assumptions are clarified, we solve for $\pi_{\Delta}$ in the next subsection.

\subsection{Stationary Distribution}\label{sec-SteadyState}
Let $ET$ be the expected transmission time of an update. Since the channel remains idle if no transmission is initiated and the expected transmission time of an update is $ET$, $\pi_{\Delta}$ satisfies the following equation.
\begin{equation}\label{eq-TotalProbs}
\sum_{\Delta=0}^{\tau-1}\pi_\Delta + ET\sum_{\Delta=\tau}^{\infty}\pi_\Delta = 1,
\end{equation}
where $ET= \sum_{t=1}^{t_{max}}tp_t$ under \textbf{Assumption 1} and $ET= \sum_{t=1}^{t_{max}}tp_t+t_{max}p_{t^+}$ under \textbf{Assumption 2}. We notice that there is still infinitely many $\pi_{\Delta}$ to calculate. To overcome the infinity, we recall that, under threshold policy, the suggested action is $a=1$ for all the state $s=(\Delta,0,-1)$ with $\Delta\geq\tau$. Hence, we define $\Pi\triangleq \sum_{\Delta=\omega}^{\infty}\pi_\Delta$ where $\omega \triangleq t_{max} + \tau+1$. As we will see in the following subsections, $\Pi$ and $\pi_{\Delta}$ for $0\leq\Delta<\omega-1$ are sufficient for calculating the expected AoII achieved by the threshold policy. With $\Pi$ in mind, we have the following theorem.
\begin{theorem}\label{prop-StationaryDistribution}
For $0<\tau<\infty$, $\Pi$ and $\pi_{\Delta}$ for $0\leq\Delta<\omega-1$ are the solution to the following system of linear equations.
\[
\pi_0 = (1-p)\pi_0 + p\sum_{i=1}^{\tau-1}\pi_i+ P_{1,0}(1)\left(\sum_{i=\tau}^{\omega-1}\pi_i+\Pi\right).
\]
\[
\pi_1= p\pi_0 + P_{1,1}(1)\left(\sum_{i=\tau}^{\omega-1}\pi_i+\Pi\right).
\]
\[
\Pi = \sum_{i=\tau+1}^{\omega-1}\left(\sum_{k=\tau+1}^iP_{i,t_{max}+k}(1)\right)\pi_i + \sum_{i=1}^{t_{max}}\bigg(P_{\omega,\omega+i}(1)\bigg)\Pi.
\]
\[
\sum_{i=0}^{\tau-1}\pi_i + ET\left(\sum_{i=\tau}^{\omega-1}\pi_i+\Pi\right) = 1.
\]
For each $2\leq\Delta\leq t_{max}-1$,
\begin{multline*}
\pi_\Delta =\\
\begin{dcases}
(1-p)\pi_{\Delta-1} + P_{\tau,\Delta}(1)\left(\sum_{i=\tau}^{\omega-1}\pi_i+\Pi\right) & \Delta-1<\tau,\\
\sum_{i=\tau}^{\Delta-1}P_{i,\Delta}(1)\pi_i + P_{\Delta,\Delta}(1)\left(\sum_{i=\Delta}^{\omega-1}\pi_i+\Pi\right) & \Delta-1\geq\tau.
\end{dcases}
\end{multline*}
For each $t_{max}\leq\Delta\leq\omega-1$,
\[
\pi_{\Delta} = \begin{dcases}
(1-p)\pi_{\Delta-1} & \Delta-1<\tau,\\
\sum_{i=\tau}^{\Delta-1}P_{i,\Delta}(1)\pi_i & \Delta-1\geq\tau.
\end{dcases}
\]
\end{theorem}
\begin{proof}
We delve into the definition of $\Pi$. By leveraging the structural property of the threshold policy and the properties of $P_{\Delta,\Delta'}(a)$, we obtain the above system of linear equations. The complete proof can be found in Appendix \ref{pf-StationaryDistribution} of the supplementary material.
\end{proof}
\begin{remark}
The size of the system of linear equations detailed in Theorem \ref{prop-StationaryDistribution} is $\omega+1$.
\end{remark}
\begin{corollary}\label{cor-AoIISpecialCase1}
When $\tau=0$,
\[
\pi_{0} = \frac{P_{1,0}(1)}{ET[1-P_{0,0}(1)+P_{1,0}(1)]}.
\]
For each $1\leq\Delta\leq t_{max}$,
\[
\pi_{\Delta} = \sum_{i=0}^{\Delta-1}P_{i,\Delta}(1)\pi_i + P_{\Delta,\Delta}(1)\left(\frac{1}{ET} - \sum_{i=0}^{\Delta-1}\pi_i\right).
\]
\[
\Pi = \ddfrac{\sum_{i=1}^{t_{max}}\left(\sum_{k=1}^{i}P_{i,t_{max}+k}(1)\right)\pi_i}{1-\sum_{i=1}^{t_{max}}P_{t_{max}+1,t_{max}+1+i}(1)}.
\]

When $\tau=1$,
\[
\pi_0 = \frac{P_{1,0}(1)}{pET+P_{1,0}(1)},\quad \pi_1 = \frac{pP_{1,0}(1)+pP_{1,1}(1)}{pET+P_{1,0}(1)}.
\]
For each $2\leq \Delta\leq t_{max}+1$,
\[
\pi_\Delta = \sum_{i=1}^{\Delta-1}P_{i,\Delta}(1)\pi_i + P_{\Delta,\Delta}(1)\left(\frac{1-\pi_0}{ET} - \sum_{i=1}^{\Delta-1}\pi_i\right).
\]
\[
\Pi = \ddfrac{\sum_{i=2}^{t_{max}+1}\left(\sum_{k=2}^iP_{i,t_{max}+k}(1)\right)\pi_i}{1-\sum_{i=1}^{t_{max}}P_{t_{max}+2,t_{max}+2+i}(1)}.
\]
\end{corollary}
\begin{proof}
The calculations follow similar steps as detailed in the proof of Theorem \ref{prop-StationaryDistribution}. The complete proof can be found in Appendix \ref{pf-AoIISpecialCase1} of the supplementary material.
\end{proof}
We will calculate the expected AoII in the next subsection based on the above results.

\subsection{Expected AoII}
Let $\bar{\Delta}_{\tau}$ be the expected AoII achieved by threshold policy $\tau$. Then,
\begin{equation}\label{eq-expectedAoIIDet}
\bar{\Delta}_{\tau} = \sum_{\Delta=0}^{\tau-1}C(\Delta,0)\pi_\Delta + \sum_{\Delta=\tau}^{\infty}C(\Delta,1)\pi_\Delta,
\end{equation}
where $C(\Delta,a)$ is the expected sum of AoII during the transmission of the update caused by the operation of $a$ at state $s=(\Delta,0,-1)$. Note that $C(\Delta,a)$ includes the AoII for being at state $s=(\Delta,0,-1)$.
\begin{remark}\label{rem-CandP}
In order to have a more intuitive understanding of the definition of  $C(\Delta,a)$, we use $\eta$ to denote a possible path of the state during the transmission of the update and let $H$ be the set of all possible paths. Moreover, we denote by $C_{\eta}$ and $P_{\eta}$ the sum of AoII and the probability associated with path $\eta$, respectively. Then,
\[
C(\Delta,a) = \sum_{\eta\in H}P_{\eta}C_{\eta}.
\]
For example, we consider the case of $p_2=1$, where the transmission takes $2$ time slots to be delivered. Also, action $a=1$ is taken at state $(2,0,-1)$. Then, a sample path $\eta$ of the state during the transmission can be the following.
\[
(2,0,-1)\rightarrow(3,1,1)\rightarrow(4,0,-1).
\]
By our definition, $C_{\eta}=2+3=5$ and $P_{\eta} = Pr[(3,1,1)\mid (2,0,-1),a=1]\cdot Pr[(4,0,-1)\mid (3,1,1),a=1]$ for the above sample path.
\end{remark}
In the following, we calculate $C(\Delta,a)$. Similar to Section \ref{sec-StateTransProb}, we define $C^{t}(\Delta,a)$ as the expected sum of AoII during the transmission of the update caused by action $a$ at state $s=(\Delta,0,-1)$, given that the transmission takes $t$ time slots. Then, under \textbf{Assumption 1},
\begin{equation}\label{eq-CompactCostRan}
C(\Delta,a) = \begin{dcases}
\Delta & a=0,\\
\sum_{t=1}^{t_{max}}p_tC^{t}(\Delta,1) & a=1,
\end{dcases}
\end{equation}
and, under \textbf{Assumption 2},
\begin{equation}\label{eq-CompactCostRan2}
C(\Delta,a) = \begin{dcases}
\Delta & a=0,\\
\sum_{t=1}^{t_{max}}p_tC^{t}(\Delta,1) + p_{t^+}C^{t_{max}}(\Delta,1) & a=1.
\end{dcases}
\end{equation}
Hence, obtaining the expressions of $C^{t}(\Delta,1)$ is sufficient. To this end, we define $C^k(\Delta)$ as the expected AoII $k$ time slots after the transmission starts at state $s=(\Delta,0,-1)$, given that the transmission is still in progress. Then, we have the following lemma.
\begin{lemma}\label{lem-CompactCost}
$C^{t}(\Delta,1)$ is given by
\[
C^{t}(\Delta,1) = \sum_{k=0}^{t-1}C^k(\Delta),
\]
where $C^k(\Delta)$ is given by \eqref{eq-Cost}.
\begin{figure*}[!t]
\normalsize
\begin{equation}\label{eq-Cost}
C^k(\Delta) = \begin{dcases}
\sum_{h=1}^{k} hp^{(k-h)}p(1-p)^{h-1} & \Delta=0,\\
\sum_{h=1}^{k-1} h(1-p^{(k-h)})p(1-p)^{h-1} + (\Delta+k)(1-p)^k & \Delta>0.
\end{dcases}
\end{equation}
\hrulefill
\vspace*{4pt}
\end{figure*}
\end{lemma}
\begin{proof}
The expression of $C^k(\Delta)$ is obtained by analyzing the system dynamics. The complete proof can be found in Appendix \ref{pf-ExpectedCost} of the supplementary material.
\end{proof}
Next, we calculate the expected AoII achieved by the threshold policy. We start with the case of $\tau=\infty$.
\begin{theorem}\label{prop-LazyPerformance}
The expected AoII achieved by the threshold policy with $\tau=\infty$ is
\[
\bar{\Delta}_\infty = \frac{1}{2p}.
\]
\end{theorem}
\begin{proof}
In this case, the transmitter will never initiate any transmissions. Hence, the state transitions are straightforward. The complete proof can be found in Appendix \ref{pf-LazyPerformance} of the supplementary material.
\end{proof}
In the following, we focus on the case where $\tau$ is finite. We recall that the expected AoII is given by \eqref{eq-expectedAoIIDet}. The problem arises because of the infinite sum. To overcome this, we adopt a similar approach as proposed in Section \ref{sec-SteadyState}. More precisely, we leverage the structural property of the threshold policy and define $\Sigma\triangleq \sum_{\Delta=\omega}^{\infty}C(\Delta,1)\pi_\Delta$. Then, equation \eqref{eq-expectedAoIIDet} can be written as
\[
\bar{\Delta}_{\tau} = \sum_{i=0}^{\tau-1}C(i,0)\pi_i + \sum_{i=\tau}^{\omega-1}C(i,1)\pi_i + \Sigma.
\]
As we have obtained the expressions of $\pi_\Delta$ and $C(\Delta,a)$ in previous subsections, it is sufficient to obtain the expression of $\Sigma$.
\begin{theorem}\label{prop-Performance}
Under \textbf{Assumption 1} and for $0\leq\tau<\infty$,
\[
\Sigma = \ddfrac{\sum_{t=1}^{t_{max}}\left[p_tP^t_{1,1+t}(1)\left(\sum_{i=\omega-t}^{\omega-1}C(i,1)\pi_i\right) + \Delta_t'\Pi_t\right]}{1-\sum_{t=1}^{t_{max}}\bigg(p_tP^t_{1,1+t}(1)\bigg)},
\]
where
\[
\Pi_t= p_{t}P^{t}_{1,1+t}(1)\left(\sum_{i=\omega-t}^{\omega-1}\pi_i + \Pi\right),
\]
\[
\Delta_t' = \sum_{i=1}^{t_{max}}p_i\left(\frac{t-t(1-p)^i}{p}\right).
\]
\end{theorem}
\begin{proof}
We delve into the definition of $\Sigma$ and repeatedly use the properties of $C(\Delta,a)$ and $P_{\Delta,\Delta'}(a)$. The complete proof can be found in Appendix \ref{pf-Performance} of the supplementary material.
\end{proof}
\begin{theorem}\label{prop-ThresholdPerformance}
Under \textbf{Assumption 2} and for $0\leq\tau<\infty$,
\[
\Sigma = \ddfrac{\sum_{t=1}^{t_{max}}\left[\left(\sum_{i=\omega-t}^{\omega-1}\Upsilon(i+t,t)C(i,1)\pi_i\right) + \Delta_t'\Pi_t\right]}{1-\sum_{t=1}^{t_{max}}\Upsilon(\omega+t,t)},
\]
where
\[
\Upsilon(\Delta,t)= p_tP^t_{\Delta-t,\Delta}(1) + p_{t^+}P^{t^+}_{\Delta-t,\Delta}(1),
\]
\[
\Pi_t = \sum_{i=\omega-t}^{\omega-1}\Upsilon(i+t,t)\pi_i + \Upsilon(\omega+t,t)\Pi,
\]
\[
\Delta_t' = \sum_{i=1}^{t_{max}}p_i\left(\frac{t-t(1-p)^i}{p}\right) + p_{t^+}\left(\frac{t-t(1-p)^{t_{max}}}{p}\right).
\]
\end{theorem}
\begin{proof}
The proof is similar to that of Theorem \ref{prop-Performance}. The complete proof can be found in Appendix \ref{pf-ThresholdPerformance} of the supplementary material.
\end{proof}

\section{Numerical Results}\label{sec-Numerical}
We lay out the numerical results in Fig.~\ref{fig-Numerical}.
\begin{figure*}%
\centering
\begin{subfigure}{2in}
\centering
\includegraphics[width=2in]{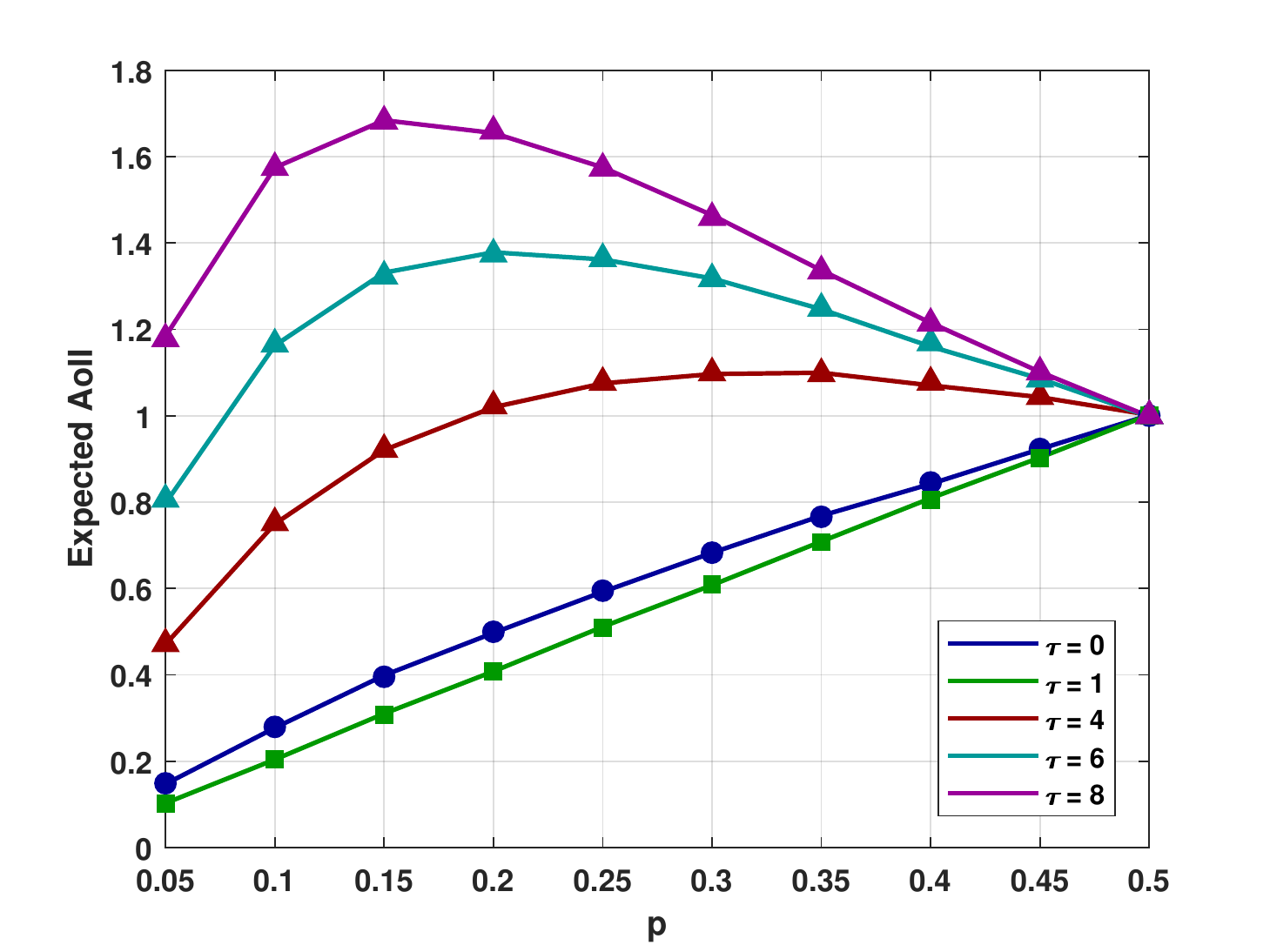}%
\caption{Performance under \textbf{Assumption 1} and Geometric distribution.}%
\label{fig-Assumption1Geo}%
\end{subfigure}\hfill%
\begin{subfigure}{2in}
\centering
\includegraphics[width=2in]{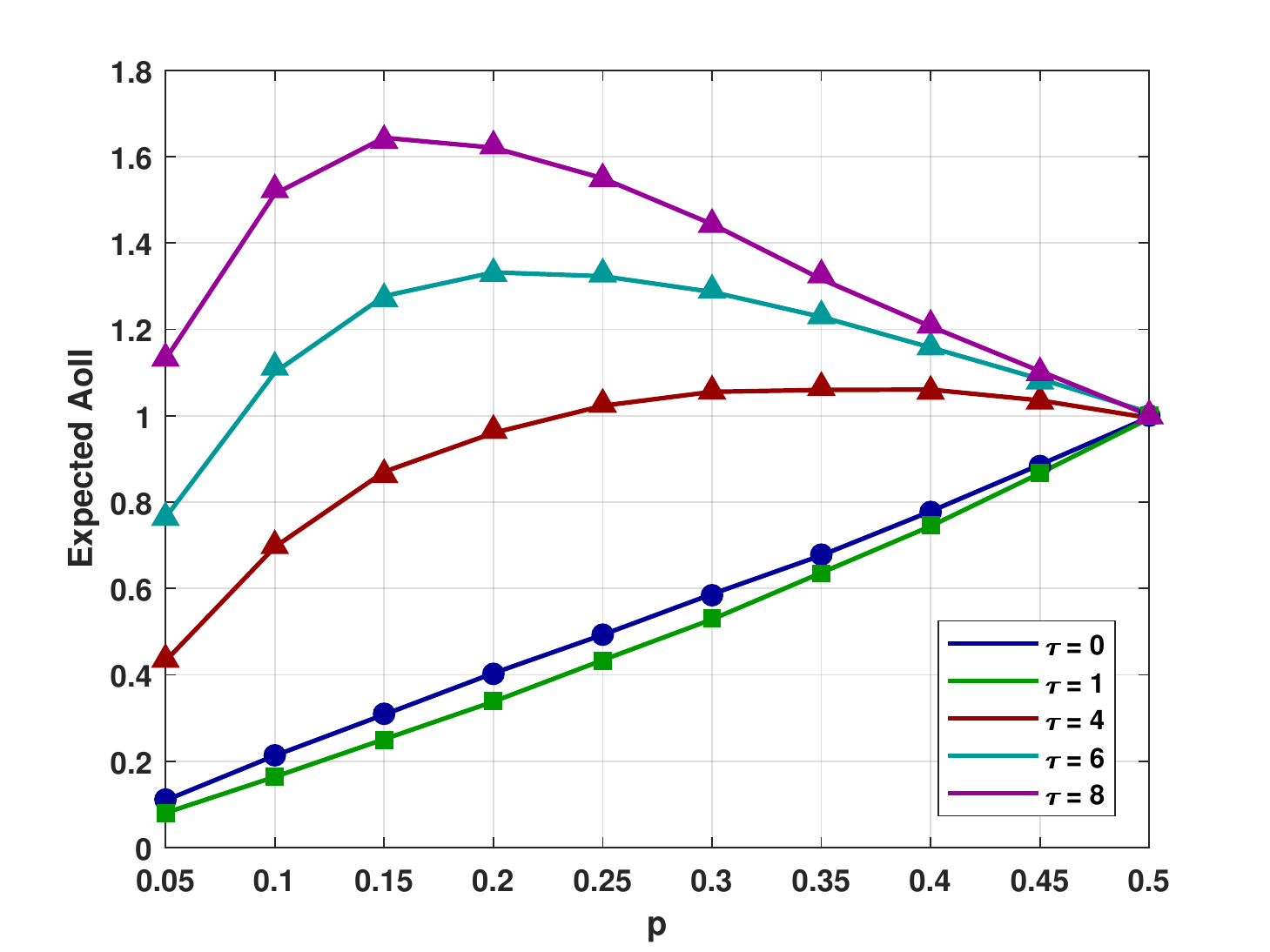}%
\caption{Performance under \textbf{Assumption 1} and Zipf distribution.}%
\label{fig-Assumption1Zipf}%
\end{subfigure}\hfill%
\begin{subfigure}{2in}
\centering
\includegraphics[width=2in]{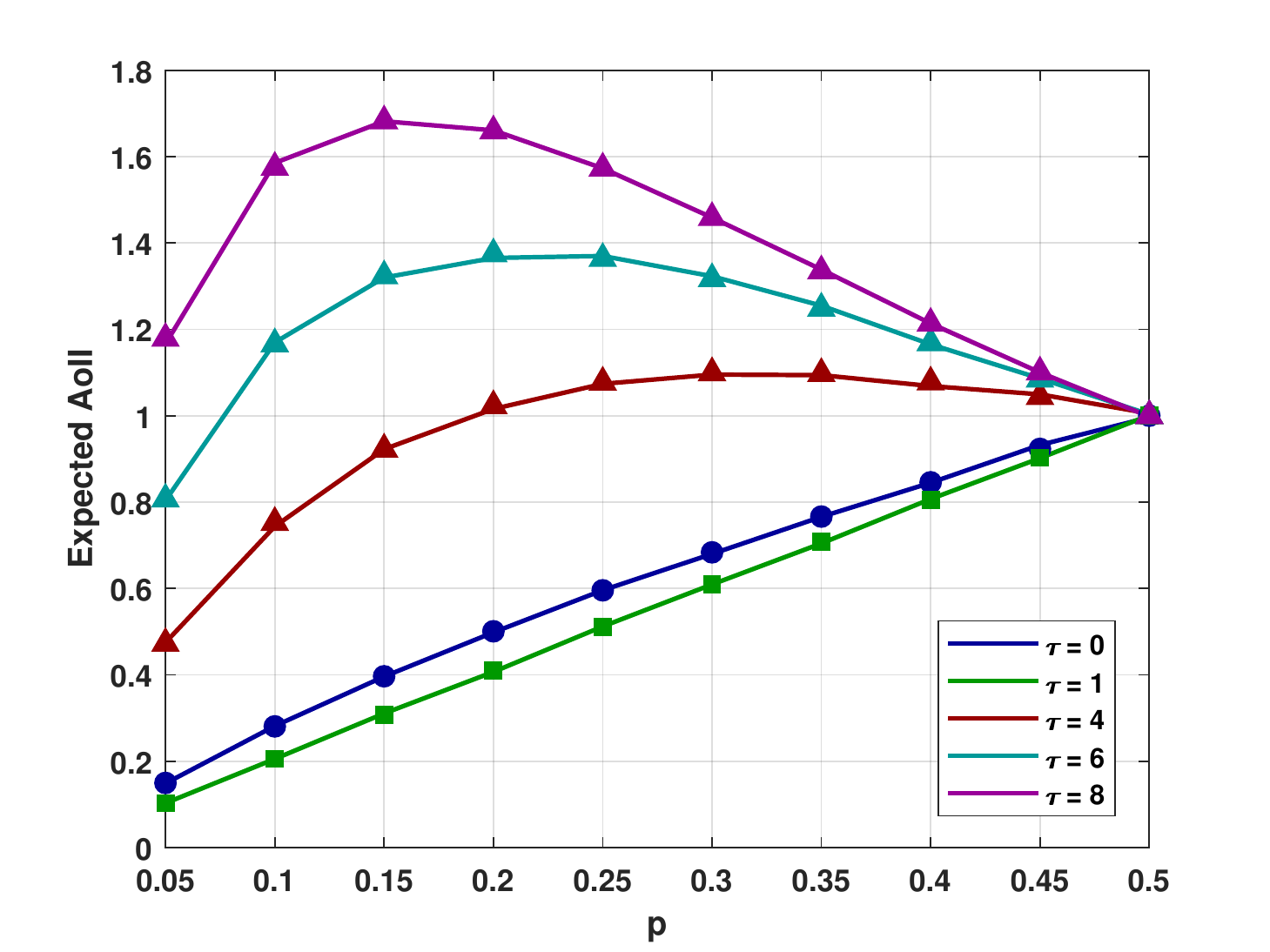}%
\caption{Performance under \textbf{Assumption 2} and Geometric distribution.}%
\label{fig-Assumption2Geo}%
\end{subfigure}%
\caption{Illustrations of the expected AoII in the function of $p$ and $\tau$. In the figure, lines represent the simulation results, and markers represent the calculation results. We set the upper limit on the transmission time $t_{max}=5$, the success probability in Geometric distribution $p_s = 0.7$, and the constant in Zipf distribution $a=3$. The simulation results are the average of $15$ runs, each run containing $25000$ epochs.}
\label{fig-Numerical}
\end{figure*}
\begin{itemize}
\item Fig.~\ref{fig-Assumption1Geo} considers the system under \textbf{Assumption 1} with $t_{max}=5$ and Geometric transmission delay with success probability $p_s=0.7$. Specifically, $p_t = (1-p_s)^{t-1}p_s$. For each considered threshold $\tau$, we vary the value of $p$ and plot the expected AoIIs obtained through numerical simulation and calculated using Section \ref{sec-MDP}.
\item Fig.~\ref{fig-Assumption1Zipf} considers the same system as the one considered in Fig.~\ref{fig-Assumption1Geo}, except this time, the transmission delay follows a Zipf distribution with the constant $a=3$. More precisely, $p_t = \frac{t^{-a}}{\sum_{i=1}^{t_{max}}i^{-a}},\ 1\leq t\leq t_{max}$. Zipf distribution is commonly used in linguistics and is also studied in the context of AoII~\cite{b6} and AoI~\cite{b11}.
\item Fig.~\ref{fig-Assumption2Geo} considers the same system as the one considered in Fig.~\ref{fig-Assumption1Geo}, except this time, the system adopts \textbf{Assumption 2} with $t_{max}=5$.
\end{itemize}
We can confirm the correctness of our theoretical analysis and calculations based on the results shown in the figure. At the same time, we can conclude from the figure that the threshold policy with $\tau=0$ is not optimal. That is, when AoII is zero, transmitting new updates is harmful to the system. One reason is that updates transmitted when AoII is zero do not provide new information to the receiver. Meanwhile, the transmission will occupy the channel for a few time slots. Therefore, such action deprives the transmitter of the ability to send new updates for the next few time slots without providing the receiver with any new information.

\section{Conclusion}
In this paper, we study the AoII in a communication system where the communication channel has a random delay. To ease the theoretical analysis, we make two independent assumptions. Under the first assumption, the transmission time is capped, and the update will always be delivered. Under the second assumption, the system will terminate the transmission when its duration exceeds a specific value. For the remote receiver to know the dynamic source well, the transmitter will choose when to send updates. This paper considers the threshold policy, under which the transmitter initiates transmission only when the AoII exceeds the threshold. Leveraging the notion of the Markov chain, we analyze the system dynamics and calculate the expected AoII achieved by the threshold policy precisely.

Finally, by analyzing the calculation results, we find that the threshold policy with $\tau=0$ is not optimal. In other words, to achieve the best performance, the transmitter should selectively transmit updates. Aware of this phenomenon, we naturally ask what the optimal policy is. Also, we remember that one of the reasons for this phenomenon is that the transmitter does not have the ability to preempt. Consequently, the transmitter can only wait until the current transmission is completed before initiating a new one. Therefore, we ask what is the optimal policy if the transmitter is preemptive. These two questions will be the focus of our subsequent work. We also plan to extend the methodology demonstrated in this paper to more generic system setups, such as more complex source processes, more general AoIIs, and systems where multiple transmitter-receiver pairs exist.

\bibliographystyle{IEEEtran}
\bibliography{mybib}

\newpage

\setcounter{page}{1}

\twocolumn[
\begin{@twocolumnfalse}
\begin{center}
\Huge Supplementary Material for the Paper "Analysis of Age of Incorrect Information under Generic Transmission Delay"
\end{center}
\end{@twocolumnfalse}
\vspace{3em}]

\appendices
\section{Proof of Lemma \ref{lem-CompactTrans}}\label{pf-CompactTrans}
We recall that $P^t_{\Delta,\Delta'}(1)$ is the probability that action $a$ at state $s=(\Delta,0,-1)$ will lead to state $s'=(\Delta',0,-1)$, given that the transmission takes $t$ time slots. With this in mind, we first distinguish between different values of $\Delta$.
\begin{itemize}
\item When $\Delta=0$, the transmitted update is the same as the receiver's estimate. Hence, the receiver's estimate will not change due to receiving the transmitted update. Moreover, as indicated by \eqref{eq-AoIIDynamic}, AoII will either increases by one or decreases to zero. Hence, $\Delta'\in\{0,1,...,t\}$. Then, we further distinguish our discussion into the following cases.
\begin{itemize}
\item $\Delta'=0$ happens when the receiver's estimate is correct as a result of receiving the update. Hence, the probability of this happening is $p^{(t)}$.
\item $\Delta'=k\in\{1,...,t\}$ happens when the receiver's estimate is correct at $(t-k)$th time slot after the transmission, which happens with probability $p^{(t-k)}$. Then, the estimate remains incorrect for the remainder of the transmission time. This happens when the source first changes state, then remains in the same state throughout the rest of the transmission. Hence, the probability of this happening is $p(1-p)^{k-1}$. Combining together, $\Delta'=k$ happens with probability $p^{(t-k)}p(1-p)^{k-1}$.
\end{itemize}
Combining together, we have
\[
P^{t}_{0,\Delta'}(1) = 
\begin{dcases}
p^{(t)} & \Delta'=0,\\
p^{(t-k)}p(1-p)^{k-1} & 1\leq\Delta'= k\leq t,\\
0 & otherwise.
\end{dcases}
\]
\item When $\Delta>0$, the transmitted update is different from the receiver's estimate. Hence, the receiver's estimate will flip as a result of receiving the transmitted update. According to \eqref{eq-AoIIDynamic}, we know $\Delta'\in\{0,1,...,t-1,\Delta+t\}$. Hence, we further distinguish between the following cases.
\begin{itemize}
\item $\Delta'=0$ happens in the same case as discussed in the case of $\Delta=0$. Hence, the estimate is correct with probability $p^{(t)}$.
\item $\Delta'=1$ happens when the estimate is correct at $(t-1)$th time slot after the transmission, which happens with probability $1-p^{(t-1)}$. Then, the estimate becomes incorrect as a result of receiving the update. Since the estimate flips upon the arrival of the transmitted update, it happens when the source remains in the same state. Hence, the probability of this happening is $1-p$. Combing together, $\Delta'=1$ happens with probability $(1-p^{(t-1)})(1-p)$.
\item $\Delta'=k\in\{2,...,t-1\}$ happens when the estimate is correct at $(t-k)$th time slot after the transmission, which happens with probability $1-p^{(t-k)}$. Then, the estimate remains incorrect for the remainder of the transmission time. This happens when the dynamic source behaves the following way during the remaining transmission time. The dynamic source should first change state, then remain in the same state, and finally, change state again when the update arrives. This happens with probability $p^2(1-p)^{k-2}$. Hence, $\Delta'=k$ happens with probability $(1-p^{(t-k)})p^{2}(1-p)^{k-2}$.
\item $\Delta'=\Delta+t$ happens when the estimate is incorrect throughout the transmission. Since the estimate will flip when the update is received, this happens when the source stays in the same state until the update arrives. Hence, $\Delta'=\Delta+t$ happens with probability $p(1-p)^{t-1}$.
\end{itemize}
Combining together, for $\Delta>0$, we have
\begin{multline*}
P^{t}_{\Delta,\Delta'}(1) = \\
\begin{dcases}
p^{(t)} & \Delta'=0,\\
(1-p^{(t-1)})(1-p) & \Delta'=1,\\
(1-p^{(t-k)})p^{2}(1-p)^{k-2} & 2\leq\Delta'=k\leq t-1,\\
p(1-p)^{t-1} & \Delta'=\Delta+t,\\
0 & otherwise.
\end{dcases}
\end{multline*}
\end{itemize}

\section{Proof of Corollary \ref{lem-StateTransProb}}\label{pf-StateTransProb}
We first lay out some properties of $P^{t}_{\Delta,\Delta'}(1)$ under \textbf{Assumption 1}. By analyzing the expressions presented in Lemma \ref{lem-CompactTrans}, we can easily conclude that, for any feasible $t$, $P^{t}_{\Delta,\Delta'}(1)$ possesses the following properties.
\begin{itemize}
\item $P^{t}_{\Delta,0}(1)$ and $P^{t}_{\Delta,\Delta+t}(1)$ are both independent of $\Delta$.
\item $P^{t}_{\Delta,\Delta'}(1)$ is independent of $\Delta$ when $\Delta>0$ and $0\leq\Delta'\leq t-1$.
\item $P^{t}_{\Delta,\Delta'}(1) = 0$ when $\Delta'>\Delta+t$ or when $t-1<\Delta'<\Delta+t$.
\end{itemize}
Leveraging the above properties, the equivalent expression in the corollary can be obtained easily. In the following, we focus on proving the properties.
\begin{itemize}
\item \textbf{property 1}: When $\Delta'=0$, $P_{\Delta,0}(1) = \sum_{t=1}^{t_{max}}p_tP^t_{\Delta,0}(1)$ for any $\Delta\geq0$. Since $P^{t}_{\Delta,0}(1)$ is independent of $\Delta$, property 1 holds in this case. Then, we consider the case of $1\leq\Delta'\leq t_{max}-1$ and $\Delta\geq\Delta'$. In this case,
\[
P_{\Delta,\Delta'}(1) = \sum_{t=\Delta'}^{t_{max}}p_tP^t_{\Delta,\Delta'}(1),
\]
where $P^t_{\Delta,\Delta'}(1)$ is independent of $\Delta$. Hence, $P_{\Delta,\Delta'}(1)$ is independent of $\Delta$. Combining together, property 1 holds.
\item \textbf{property 2}: We notice that, when $\Delta'\geq t_{max}$,
\[
P_{\Delta,\Delta'}(1) = p_{t'}P^{t'}_{\Delta,\Delta'}(1) = p_{t'}P^{t'}_{\Delta,\Delta+t'}(1).
\]
We recall that $P^{t'}_{\Delta,\Delta+t'}(1)$ is independent of $\Delta$. Then, we can conclude that $P_{\Delta,\Delta'}(1)$ depends only on $t'$. Thus, property 2 holds.
\item \textbf{property 3}: The equivalent expression in corollary indicates that the property holds when $\Delta'>\Delta+t_{max}$. In the case of $t_{max}-1<\Delta'<\Delta+1$, we have
\[
P_{\Delta,\Delta'}(1) = p_{t'}P^{t'}_{\Delta,\Delta'}(1),
\]
where $t'\leq0$. By definition, $P_{\Delta,\Delta'}(1)=0$. Hence, property 3 holds.
\end{itemize}

\section{Proof of Lemma \ref{lem-Case2TransProb}}\label{pf-Case2TransProb}
The proof is similar to that of Lemma \ref{lem-CompactTrans} and Corollary \ref{lem-StateTransProb}. We first derive the expressions of $P^t_{\Delta,\Delta'}(1)$ and $P^{t^{+}}_{\Delta,\Delta'}(1)$. To this end, we start with the case of $\Delta=0$. In this case, the transmitted update is the same as the receiver's estimate. With this in mind, we distinguish between different values of $t$.
\begin{itemize}
\item When $1\leq t<t_{max}$, the update is delivered after $t$ time slot. Hence, $\Delta'\in\{0,1,...,t\}$. Then, we further distinguish between different values of $\Delta'$.
\begin{itemize}
\item $\Delta'=0$ in the case where the receiver's estimate is correct when the update is delivered. Hence, $\Delta'=0$ happens with probability $p^{(t)}$.
\item $\Delta' = k \in\{1,2,...,t\}$ when the receiver's estimate is correct at the $(t-k)$th time slots after the transmission occurs. Then, the source flips the state and remains in the same state for the remainder of the transmission. Hence, $\Delta' = k \in\{1,2,...,t\}$ happens with probability $p^{(t-k)}p(1-p)^{k-1}$.
\end{itemize}
\item When $t = t_{max}$, the update either arrives or be discarded. In this case, $\Delta'\in\{0,1,...,t_{max}\}$. We recall that the update is the same as the receiver's estimate. Hence, the receiver's estimate will not change in both cases. Consequently, $P^{t_{max}}_{0,\Delta'}(1)=P^{t^+}_{0,\Delta'}(1)$, which can be obtained by setting the $t$ in the above case to $t_{max}$.
\end{itemize}
Combining together, for each $1\leq t\leq t_{max}$,
\[
P^{t}_{0,\Delta'}(1) = 
\begin{dcases}
p^{(t)} & \Delta'=0,\\
p^{(t-k)}p(1-p)^{k-1} & 1\leq\Delta'= k\leq t,\\
0 & otherwise.
\end{dcases}
\]
\[
P^{t^+}_{0,\Delta'}(1) = P^{t_{max}}_{0,\Delta'}(1).
\]
Then, we consider the case of $\Delta>0$. We notice that, in this case, the receiver's estimate will flip upon receiving the update. Then, we distinguish between different values of $t$.
\begin{itemize}
\item When $1\leq t<t_{max}$, the update is delivered after $t$ time slots, and the receiver's estimate will flip. Hence, $\Delta'\in\{0,1,...,t-1,\Delta+t\}$. Then, we further distinguish between different values of $\Delta'$.
\begin{itemize}
\item $\Delta'=0$ in the case where the receiver's estimate is correct when the update is received. Hence, $\Delta'=0$ happens with probability $p^{(t)}$.
\item $\Delta'=1$ when the receiver's estimate is correct at $(t-1)$th time slot after the transmission starts and becomes incorrect when the update arrives. Hence, $\Delta'=1$ happens with probability $(1-p^{(t-1)})(1-p)$.
\item $\Delta'=k\in\{2,3,...,t-1\}$ when the receiver's estimate is correct at $(t-k)$th time slot after the transmission starts. Then, the source changes state and remains in the same state. Finally, at the time slot when the update arrives, the source flips state again. Hence, $\Delta'=k\in\{2,3,...,t-1\}$ happens with probability $(1-p^{(t-k)})p^{2}(1-p)^{k-2}$.
\item $\Delta'=\Delta+t$ when the estimate is incorrect throughout the transmission. We recall that the receiver's estimate will flip when the update arrives. Hence, $\Delta'=\Delta+t$ when the source remains in the same state until the update arrives, which happens with probability $p(1-p)^{t-1}$.
\end{itemize}
\item When $t=t_{max}$ and the transmitted update is delivered, the receiver's estimate flips. In this case, $\Delta'\in\{0,1,...,t_{max}-1,\Delta+t_{max}\}$. Hence, $P^{t_{max}}_{\Delta,\Delta'}(1)$ can be obtained by setting the $t$ in the above case to $t_{max}$.
\item When $t=t_{max}$ and the transmitted update is discarded, the receiver's estimate remains the same. In this case, $\Delta'\in\{0,1,...,t_{max}-1,\Delta+t_{max}\}$. Then, we further divide our discussion into the following cases.
\begin{itemize}
\item $\Delta'=0$ when the receiver's estimate is correct at the $t_{max}$th time slot after the transmission starts, which happens when the state of the source at the time slot the update is discarded is different from that when the transmission started. Hence, $\Delta'=0$ happens with probability $1- p^{(t_{max})}$.
\item $\Delta' = k \in\{1,2,...,t_{max}-1\}$ when the receiver's estimate is correct at $(t_{max}-k)$th time slot after the transmission starts. Then, the source changes state and remains in the same state for the remainder of the transmission. Hence, $\Delta'=k\in\{1,2,...,t_{max}-1\}$ happens with probability $(1-p^{(t_{max}-k)})p(1-p)^{k-1}$.
\item $\Delta' = \Delta+t_{max}$ when the source remains in the same state throughout the transmission. Combining with the source dynamic, we can conclude that $\Delta' = \Delta+t_{max}$ happens with probability $(1-p)^{t_{max}}$.
\end{itemize}
\end{itemize}
Combining together, for $\Delta>0$ and each $1\leq t\leq t_{max}$,
\begin{multline*}
P^{t}_{\Delta,\Delta'}(1) = \\
\begin{dcases}
p^{(t)} & \Delta'=0,\\
(1-p^{(t-1)})(1-p) & \Delta'=1,\\
(1-p^{(t-k)})p^{2}(1-p)^{k-2} & 2\leq\Delta'=k\leq t-1,\\
p(1-p)^{t-1} & \Delta'=\Delta+t,\\
0 & otherwise.
\end{dcases}
\end{multline*}
\begin{multline*}
P^{t^+}_{\Delta,\Delta'}(1) = \\
\begin{dcases}
1-p^{(t_{max})} & \Delta'=0,\\
(1-p^{(t_{max}-k)})p(1-p)^{k-1} & 1\leq \Delta'= k\leq t_{max}-1,\\
(1-p)^{t_{max}} & \Delta' = \Delta+t_{max},\\
0 & otherwise.
\end{dcases}
\end{multline*}
By analyzing the expressions, we can easily conclude that $P^t_{\Delta,\Delta'}(1)$ and $P^{t^{+}}_{\Delta,\Delta'}(1)$ possess the following properties.
\begin{itemize}
\item$P^{t}_{\Delta,\Delta+t}(1)$ and $P^{t^+}_{\Delta,\Delta+t_{max}}(1)$ are independent of $\Delta$ when $\Delta>0$.
\item $P^t_{\Delta,\Delta'}(1)$ is independent of $\Delta$ when $\Delta>0$ and $0\leq\Delta'\leq t-1$.
\item $P^t_{\Delta,\Delta'}(1)=0$ when $\Delta>0$ and $t-1<\Delta'< \Delta+t$.
\item $P^{t^+}_{\Delta,\Delta'}(1)$ is independent of $\Delta$ when $\Delta>0$ and $0\leq\Delta'\leq t_{max}-1$.
\item $P^{t^+}_{\Delta,\Delta'}(1)=0$ when $\Delta>0$ and $t_{max}-1<\Delta'< \Delta+t_{max}$.
\end{itemize}
Leveraging the properties above, we proceed with proving the second part of the lemma. The equivalent expression can be obtained easily by analyzing \eqref{eq-compactTransProbs}. Hence, the details are omitted. In the following, we focus on proving the presented properties.
\begin{itemize}
\item \textbf{property 1}: We notice that, when $0 \leq\Delta' \leq t_{max} -1$ and $\Delta\geq \max\{1,\Delta'\}$,
\[
P_{\Delta,\Delta'}(1) = \sum_{t=\Delta'}^{t_{max}}p_tP^t_{\Delta,\Delta'}(1)+ p_{t^+}P^{t^+}_{\Delta,\Delta'}(1).
\]
Then, we divide the discussion into the following two cases.
\begin{itemize}
\item $\Delta\geq\max\{1,\Delta'\}$ indicates that $\Delta>0$ and $\Delta'<\Delta+t_{max}$. Hence, $P^{t^+}_{\Delta,\Delta'}(1)$ is independent of $\Delta$.
\item $\Delta\geq\max\{1,\Delta'\}$ indicates that $\Delta>0$ and $\Delta'<\Delta+t$. Hence, $P^t_{\Delta,\Delta'}(1)$ is independent of $\Delta$ for any feasible $t$.
\end{itemize}
Combining together, we can conclude that property 1 holds.
\item \textbf{property 2}: We notice that, when $\Delta'\geq t_{max}$,
\[
P_{\Delta,\Delta'}(1) = p_{t'}P^{t'}_{\Delta,\Delta'}(1)+ p_{t^+}P^{t^+}_{\Delta,\Delta'}(1).
\]
Then, we divide the discussion into the following two cases.
\begin{itemize}
\item Since $t'=\Delta'-\Delta$, $P_{\Delta,\Delta'}^{t'}(1) = P^{t'}_{\Delta,\Delta+t'}(1)$. Then, we know that $P^{t'}_{\Delta,\Delta'}(1)$ is independent of $\Delta>0$ when $t'>0$ and $P^{t'}_{\Delta,\Delta'}(1)=0$ when $t'\leq0$ by definition. Hence, $P_{\Delta,\Delta'}^{t'}(1)$ depends on $t'$.
\item When $\Delta'\geq t_{max}$ and $\Delta'\neq\Delta+t_{max}$, $P^{t^+}_{\Delta,\Delta'}(1)=0$ for $\Delta>0$. Also, $P^{t^+}_{\Delta,\Delta'}(1)$ is independent of $\Delta>0$ when $\Delta'=\Delta+t_{max}$. Hence, $P^{t^+}_{\Delta,\Delta'}(1)$ depends only on $t'$.
\end{itemize}
Combining together, property 2 holds.
\item \textbf{property 3}: When $\Delta'>\Delta+t_{max}$, the property holds apparently. When $t_{max}-1<\Delta'<\Delta+1$, 
\[
P_{\Delta,\Delta'}(1) = p_{t'}P^{t'}_{\Delta,\Delta'}(1)+ p_{t^+}P^{t^+}_{\Delta,\Delta'}(1),
\]
where $t'\leq0$. Then, by definition, $P^{t'}_{\Delta,\Delta'}(1)=0$. Moreover, we recall that $t_{max}>1$, which indicates that $P^{t^+}_{\Delta,\Delta'}(1)=0$. Hence, property 3 holds.
\end{itemize}

\section{Proof of Theorem \ref{prop-StationaryDistribution}}\label{pf-StationaryDistribution}
We recall that $\pi_{\Delta}$ satisfies \eqref{eq-CompactBalanceEq2} and \eqref{eq-TotalProbs}. Then, plugging in the probabilities yields the following system of linear equations.
\begin{equation}\label{eq-0util15}
\begin{split}
\pi_0 = & (1-p)\pi_0 + p\sum_{i=1}^{\tau-1}\pi_i+ \sum_{i=\tau}^{\infty}P_{i,0}(1)\pi_i \\
= & (1-p)\pi_0 + p\sum_{i=1}^{\tau-1}\pi_i+ P_{1,0}(1)\sum_{i=\tau}^{\infty}\pi_i.
\end{split}
\end{equation}
\begin{equation}\label{eq-0util16}
\pi_1 = p\pi_0 + \sum_{i=\tau}^{\infty}P_{i,1}(1)\pi_i = p\pi_0 + P_{1,1}(1)\sum_{i=\tau}^{\infty}\pi_i.
\end{equation}
For each $2\leq\Delta\leq t_{max}-1$,
\begin{equation}\label{eq-0util17}
\pi_\Delta = \begin{dcases}
(1-p)\pi_{\Delta-1} + P_{\tau,\Delta}(1)\sum_{i=\tau}^{\infty}\pi_i & \Delta-1<\tau,\\
\sum_{i=\tau}^{\Delta-1}P_{i,\Delta}(1)\pi_i + P_{\Delta,\Delta}(1)\sum_{i=\Delta}^{\infty}\pi_i & \Delta-1\geq\tau.
\end{dcases}
\end{equation}
For each $t_{max}\leq\Delta\leq \omega-1$,
\[
\pi_{\Delta} = \begin{dcases}
(1-p)\pi_{\Delta-1} & \Delta-1<\tau,\\
\sum_{i=\tau}^{\Delta-1}P_{i,\Delta}(1)\pi_i & \Delta-1\geq\tau.
\end{dcases}
\]
For each $\Delta\geq\omega$,
\begin{equation}\label{eq-0util6}
\pi_{\Delta}  = \sum_{i=\Delta-t_{max}}^{\Delta-1}P_{i,\Delta}(1)\pi_i.
\end{equation}
\[
\sum_{i=0}^{\tau-1}\pi_i + ET\sum_{i=\tau}^{\infty}\pi_i = 1.
\]
Note that we can pull the state transition probabilities in \eqref{eq-0util15}, \eqref{eq-0util16}, and \eqref{eq-0util17} out of the summation due to property 1 in Corollary \ref{lem-StateTransProb} and Lemma \ref{lem-Case2TransProb}. Then, we sum \eqref{eq-0util6} over $\Delta$ from $\omega$ to $\infty$.
\begin{equation}\label{eq-0util7}
\sum_{i=\omega}^{\infty}\pi_i = \sum_{i=\omega}^{\infty}\sum_{k=i-t_{max}}^{i-1}P_{k,i}(1)\pi_k.
\end{equation}
We delve deep into the right hand side (RHS) of \eqref{eq-0util7}. To this end, we expand the first summation, which yields
\begin{align*}
RHS = & \sum_{k=\tau+1}^{\omega-1}P_{k,\omega}(1)\pi_k + \sum_{k=\tau+2}^{\omega}P_{k,\omega+1}(1)\pi_k + \cdots +\\
& \sum_{k=\omega-1}^{\omega+t_{max}-2}P_{k,\omega+t_{max}-1}(1)\pi_k +\\
& \sum_{k=\omega}^{\omega+t_{max}-1}P_{k,\omega+t_{max}}(1)\pi_k + \cdots
\end{align*}
Then, we rearrange the summation.
\begin{align*}
RHS = & P_{\tau+1,\omega}(1)\pi_{\tau+1} + \sum_{k=1}^{2}P_{\tau+2,\omega+k-1}(1)\pi_{\tau+2} + \cdots +\\
& \sum_{k=1}^{t_{max}}P_{\omega-1,\omega+k-1}(1)\pi_{\omega-1} +\\
&  \sum_{k=1}^{t_{max}}P_{\omega,\omega+k}(1)\pi_{\omega} +\sum_{k=1}^{t_{max}}P_{\omega+1,\omega+k+1}(1)\pi_{\omega+1} +\cdots
\end{align*}
Leveraging property 2 in Corollary \ref{lem-StateTransProb} and Lemma \ref{lem-Case2TransProb}, we have
\begin{align*}
RHS = &\sum_{i=\tau+1}^{\omega-1}\left(\sum_{k=\tau+1}^iP_{i,t_{max}+k}(1)\right)\pi_i +\\
&\sum_{i=1}^{t_{max}}\bigg(P_{\omega,\omega+i}(1)\bigg)\left(\sum_{k=\omega}^{\infty}\pi_k\right).
\end{align*}
We define $\Pi\triangleq \sum_{i=\omega}^{\infty}\pi_i$. Then, equation \eqref{eq-0util7} becomes the following.
\begin{equation}\label{eq-0util71}
\Pi = \sum_{i=\tau+1}^{\omega-1}\left(\sum_{k=\tau+1}^iP_{i,t_{max}+k}(1)\right)\pi_i + \sum_{i=1}^{t_{max}}\bigg(P_{\omega,\omega+i}(1)\bigg)\Pi.
\end{equation}
Finally, replacing \eqref{eq-0util6} with \eqref{eq-0util71} and applying the definition of $\Pi$ yield a system of linear equations with finite size as presented in the theorem.

\section{Proof of Corollary \ref{cor-AoIISpecialCase1}}\label{pf-AoIISpecialCase1}
We start with $\tau=0$. In this case, $\omega = t_{max}+1$ and the system of linear equations becomes to the following.
\begin{align}\label{eq-0util1}
\pi_{\Delta} = & \sum_{i=0}^{\infty}P_{i,\Delta}(1)\pi_i \\
= & \begin{dcases}
P_{0,0}(1)\pi_0 + P_{1,0}(1)\sum_{i=1}^{\infty}\pi_i & \Delta=0,\\
\sum_{i=0}^{\Delta-1}P_{i,\Delta}(1)\pi_i + P_{\Delta,\Delta}(1)\sum_{i=\Delta}^{\infty}\pi_i & 1\leq\Delta\leq t_{max}.
\end{dcases}
\end{align}
\begin{align}\label{eq-0util72}
\Pi = &\sum_{i=1}^{t_{max}}\left(\sum_{k=1}^iP_{i,t_{max}+k}(1)\right)\pi_i +\\
& \sum_{i=1}^{t_{max}}P_{t_{max}+1,t_{max}+1+i}(1)\Pi.
\end{align}
\begin{equation}\label{eq-0util2}
ET\sum_{i=0}^{\infty}\pi_i = 1.
\end{equation}
We first combine \eqref{eq-0util1} and \eqref{eq-0util2}, which yields \eqref{eq-EquivalentEq3}.
\begin{figure*}[!t]
\normalsize
\begin{equation}\label{eq-EquivalentEq3}
\pi_{\Delta} = \begin{dcases}
P_{0,0}(1)\pi_0 + P_{1,0}(1)\left(\frac{1}{ET} -\pi_0\right) & \Delta=0,\\
\sum_{i=0}^{\Delta-1}P_{i,\Delta}(1)\pi_i + P_{\Delta,\Delta}(1)\left(\frac{1}{ET} - \sum_{i=0}^{\Delta-1}\pi_i\right) & 1\leq\Delta\leq t_{max}.
\end{dcases}
\end{equation}
\hrulefill
\vspace*{4pt}
\end{figure*}
Then, we have
\[
\pi_{0} = \frac{P_{1,0}(1)}{ET[1-P_{0,0}(1)+P_{1,0}(1)]}.
\]
According to \eqref{eq-0util72}, we obtain
\[
\Pi = \ddfrac{\sum_{i=1}^{t_{max}}\left(\sum_{k=1}^{i}P_{i,t_{max}+k}(1)\right)\pi_i}{1-\sum_{i=1}^{t_{max}}P_{t_{max}+1,t_{max}+1+i}(1)}.
\]

Then, we consider the case of $\tau=1$. In this case, $\omega = t_{max}+2$ and the system of linear equations reduces to the following.
\begin{equation}\label{eq-0util8}
\pi_0 = (1-p)\pi_0 + P_{1,0}(1)\sum_{i=1}^{\infty}\pi_i.
\end{equation}
\[
\pi_1 = p\pi_0 + P_{1,1}(1)\sum_{i=1}^{\infty}\pi_i.
\]
\begin{multline*}
\pi_\Delta = \sum_{i=1}^{\Delta-1}P_{i,\Delta}(1)\pi_i + P_{\Delta,\Delta}(1)\sum_{i=\Delta}^{\infty}\pi_i,\\
2\leq \Delta\leq t_{max}-1.
\end{multline*}
\begin{equation}\label{eq-corollary1util1}
\pi_{\Delta} = \sum_{i=1}^{\Delta-1}P_{i,\Delta}(1)\pi_i,\quad t_{max}\leq\Delta\leq t_{max}+1.
\end{equation}
\begin{align}\label{eq-0util73}
\Pi = & \sum_{i=2}^{t_{max}+1}\left(\sum_{k=2}^iP_{i,t_{max}+k}(1)\right)\pi_i + \\
& \sum_{i=1}^{t_{max}}P_{t_{max}+2,t_{max}+2+i}(1)\Pi.
\end{align}
\begin{equation}\label{eq-0util9}
\pi_0 + ET\sum_{i=1}^{\infty}\pi_i = 1.
\end{equation}
We first combine \eqref{eq-0util8} and \eqref{eq-0util9}, which yields
\[
\pi_0 = (1-p)\pi_0 + P_{1,0}(1)\left(\frac{1-\pi_0}{ET}\right).
\]
Hence, we have
\[
\pi_0 = \frac{P_{1,0}(1)}{pET+P_{1,0}(1)}.
\]
Similarly,
\[
\pi_1 = \frac{pP_{1,0}(1)+pP_{1,1}(1)}{pET+P_{1,0}(1)}.
\]
For each $2\leq \Delta\leq t_{max}-1$,
\begin{equation}\label{eq-corollary1util0}
\pi_\Delta = \sum_{i=1}^{\Delta-1}P_{i,\Delta}(1)\pi_i + P_{\Delta,\Delta}(1)\left(\frac{1-\pi_0}{ET} - \sum_{i=1}^{\Delta-1}\pi_i\right).
\end{equation}
According to the property 3 in Corollary \ref{lem-StateTransProb} and Lemma \ref{lem-Case2TransProb}, we know that $P_{\Delta,\Delta}(1) = 0$ when $t_{max}\leq\Delta\leq t_{max}+1$. Hence, we can combine \eqref{eq-corollary1util1} and \eqref{eq-corollary1util0}, which yields
\begin{multline*}
\pi_\Delta = \sum_{i=1}^{\Delta-1}P_{i,\Delta}(1)\pi_i + P_{\Delta,\Delta}(1)\left(\frac{1-\pi_0}{ET} - \sum_{i=1}^{\Delta-1}\pi_i\right),\\
2\leq \Delta\leq t_{max}+1.
\end{multline*}
Finally, according to \eqref{eq-0util73}, we obtain
\[
\Pi = \ddfrac{\sum_{i=2}^{t_{max}+1}\left(\sum_{k=2}^iP_{i,t_{max}+k}(1)\right)\pi_i}{1-\sum_{i=1}^{t_{max}}P_{t_{max}+2,t_{max}+2+i}(1)}.
\]

\section{Proof of Lemma \ref{lem-CompactCost}}\label{pf-ExpectedCost}
We recall that $C^k(\Delta)$ is defined as the expected AoII $k$ time slots after the transmission starts at state $s=(\Delta,0,-1)$, given that the transmission is still in progress. With this in mind, we start with the case of $\Delta=0$. As AoII either increases by one or decreases to zero, we know $C^{k}(0)\in\{0,...,k\}$. Then, we distinguish between the following cases.
\begin{itemize}
\item $C^k(0) = 0$ when the receiver's estimate is correct $k$ time slots after the transmission starts. Since $\Delta = 0$, we can easily conclude that $C^k(0) = 0$ happens with probability $p^{(k)}$.
\item $C^k(0) = h$, where $1\leq h\leq k$, happens when the receiver's estimate is correct at the $(k-h)$th time slot after the transmission starts, then, the source flips the state and stays in the same state for the remaining $h-1$ time slots. Hence, $C^k(0) = h$, where $1\leq h\leq k$, happens with probability $p^{(k-h)}p(1 - p)^{h-1}$.
\end{itemize}
Combining together, we obtain
\[
C^k(0) = \sum_{h=1}^{k} hp^{(k-h)}p(1-p)^{h-1}.
\]
Then, we consider the case of $\Delta>0$. In this case, the transmission starts when the receiver's estimate is incorrect and $C^{k}(\Delta)\in\{0,1,...k-1,\Delta+k\}$. Then, we distinguish between the following cases.
\begin{itemize}
\item $C^k(\Delta) = 0$ when the receiver's estimate is correct at the $k$th time slot after the transmission starts, which happens with probability $(1-p^{(k)})$.
\item $C^k(\Delta) = h$, where $h\in\{1,2,...,k-1\}$, happens when the receiver's estimate is correct at the $(k-h)$th slot after the transmission starts. Then, the source flips the state and stays in the same state for the remaining $h-1$ time slots. Hence, $C^k(\Delta) = h$, where $h\in\{1,2,...,k-1\}$, happens with probability $(1-p^{(k-h)})p(1 - p)^{h-1}$.
\item $C^k(\Delta) = \Delta+k$ when the estimate at the receiver side is always wrong for $k$ time slots after the transmission starts. Since $\Delta > 0$ and the receiver's estimate will not change, $C^k(\Delta) = \Delta+k$ happens with probability $(1 - p)^k$. 
\end{itemize}
Combining together, we obtain
\begin{multline*}
C^k(\Delta) = \sum_{h=1}^{k-1} h(1-p^{(k-h)})p(1-p)^{h-1} + (\Delta+k)(1-p)^k,\\
\Delta>0.
\end{multline*}

\section{Proof of Theorem \ref{prop-LazyPerformance}}\label{pf-LazyPerformance}
We recall that when $\tau=\infty$, the transmitter will never initiate any transmissions. Hence, the receiver's estimate will never change. Without loss of generality, we assume the receiver's estimate $\hat{X}_k=0$ for all $k$. The first step in calculating the expected AoII achieved by the threshold policy with $\tau=\infty$ is to calculate the stationary distribution of the induced DTMC. To this end, we know that $\pi_{\Delta}$ satisfies the following equations.
\begin{equation}\label{eq-lazypolicy1}
\pi_0 = (1-p)\pi_0 + p\sum_{i=1}^{\infty}\pi_i.
\end{equation}
\[
\pi_1 = p\pi_0.
\]
\[
\pi_{\Delta} = (1-p)\pi_{\Delta-1},\quad \Delta\geq 2.
\]
\begin{equation}\label{eq-lazypolicy2}
\sum_{i=0}^{\infty}\pi_i = 1.
\end{equation}
Combining \eqref{eq-lazypolicy1} and \eqref{eq-lazypolicy2} yields
\[
\pi_0 = (1-p)\pi_0 + p(1-\pi_0).
\]
Hence, $\pi_0=\frac{1}{2}$. Then, we can get
\[
\pi_1 = \frac{p}{2},
\]
\[
\pi_{\Delta} = (1-p)^{\Delta-1}\pi_1 = \frac{p(1-p)^{\Delta-1}}{2},\quad \Delta\geq2.
\]
Combining together, we have
\[
\pi_0=\frac{1}{2},\quad \pi_{\Delta} = \frac{p(1-p)^{\Delta-1}}{2},\quad \Delta\geq1.
\]
Since the transmitter will never make any transmission attempts, the cost for being at state $s=(\Delta,0,-1)$ is nothing but $\Delta$ itself. Hence, the expected AoII is
\[
\bar{\Delta}_\infty = \sum_{\Delta=1}^{\infty}\Delta\frac{p(1-p)^{\Delta-1}}{2} = \frac{1}{2p}.
\]

\section{Proof of Theorem \ref{prop-Performance}}\label{pf-Performance}
We recall that, for $\Delta\geq\omega$, $\pi_\Delta$ satisfies
\begin{align*}
\pi_{\Delta} = & \sum_{i=\Delta-t_{max}}^{\Delta-1}P_{i,\Delta}(1)\pi_i\\
= & \sum_{i=1}^{t_{max}}P_{i-t_{max}+\Delta-1,\Delta}(1)\pi_{i-t_{max}+\Delta-1},\quad \Delta\geq\omega.
\end{align*}
Then, by Corollary \ref{lem-StateTransProb}, when $\Delta'\geq\omega$, $P_{\Delta,\Delta'}(1) = p_{t'}P^{t'}_{\Delta,\Delta'}(1)$ where $t' = \Delta'-\Delta$. Hence,
\begin{multline*}
\pi_{\Delta} = \sum_{i=1}^{t_{max}}p_{t_{max}+1-i}P^{t_{max}+1-i}_{i-t_{max}+\Delta-1,\Delta}(1)\pi_{i-t_{max}+\Delta-1},\\
\Delta\geq\omega.
\end{multline*}
Renaming the variables yields
\[
\pi_{\Delta} = \sum_{t=1}^{t_{max}}p_{t}P^{t}_{\Delta-t,\Delta}(1)\pi_{\Delta-t},\quad \Delta\geq\omega.
\]
To proceed, we define, for each $1\leq t\leq t_{max}$,
\begin{equation}\label{eq-0util41}
\pi_{\Delta,t} \triangleq p_{t}P^{t}_{\Delta-t,\Delta}(1)\pi_{\Delta-t},\quad\Delta\geq\omega.
\end{equation}
Note that $\sum_{t=1}^{t_{max}}\pi_{\Delta,t} = \pi_{\Delta}$. Then, for a given $1\leq t\leq t_{max}$, we multiple both side of  \eqref{eq-0util41} by $C(\Delta-t,1)$ and sum over $\Delta$ from $\omega$ to $\infty$. Hence, we have
\begin{equation}\label{eq-0util51}
\sum_{i=\omega}^{\infty}C(i-t,1)\pi_{i,t} = \sum_{i=\omega}^{\infty}C(i-t,1)p_tP^t_{i-t,i}(1)\pi_{i-t}.
\end{equation}
We define $\Delta_t' \triangleq C(\Delta,1) - C(\Delta-t,1)$ where $\Delta>t$. Then, according to \eqref{eq-CompactCostRan}, we have
\[
\Delta_t' = \sum_{i=1}^{t_{max}}p_i\bigg(C^i(\Delta,1) - C^i(\Delta-t,1)\bigg).
\]
According to Lemma \ref{lem-CompactCost}, we have \eqref{eq-EquivalentEq4} and \eqref{eq-EquivalentEq5}.
\begin{figure*}[!t]
\normalsize
\begin{equation}\label{eq-EquivalentEq4}
C^{i}(\Delta-t,1) = \Delta - t + \sum_{h=1}^{i-1}\left(\sum_{k=1}^{h-1} k(1-p^{(h-k)})p(1-p)^{k-1} + (\Delta-t+h)(1-p)^h\right).
\end{equation}
\begin{equation}\label{eq-EquivalentEq5}
C^{i}(\Delta,1) = \Delta+ \sum_{h=1}^{i-1}\left(\sum_{k=1}^{h-1} k(1-p^{(h-k)})p(1-p)^{k-1} + (\Delta+h)(1-p)^h\right).
\end{equation}
\hrulefill
\vspace*{4pt}
\end{figure*}
Subtracting the two equations yields
\begin{align*}
C^{i}(\Delta,1) - C^{i}(\Delta-t,1) = & t + \sum_{h=1}^{i-1}\bigg(t(1-p)^h\bigg) \\
= & \frac{t-t(1-p)^i}{p}.
\end{align*}
Then, we have
\[
\Delta_t' = \sum_{i=1}^{t_{max}}p_i\left(\frac{t-t(1-p)^i}{p}\right).
\]
We notice that $\Delta_t'$ is independent of $\Delta$ when $\Delta>t$. Hence, \eqref{eq-0util51} can be rewritten as
\[
\sum_{i=\omega}^{\infty}\bigg(C(i,1) - \Delta_t'\bigg)\pi_{i,t} = \sum_{i=\omega-t}^{\infty}C(i,1)p_tP^t_{i,i+t}(1)\pi_i.
\]
Then, we define $\Pi_t \triangleq \sum_{i=\omega}^{\infty}\pi_{i,t}$ and $\Sigma_t \triangleq \sum_{i=\omega}^{\infty}C(i,1)\pi_{i,t}$. We notice that $P^t_{\Delta,\Delta+t}(1)$ is independent of $\Delta$ when $\Delta>0$. Hence, we obtain 
\[
\sum_{i=\omega}^{\infty}C(i,1)\pi_{i,t} - \Delta_t'\sum_{i=\omega}^{\infty}\pi_{i,t} = p_tP^t_{1,1+t}(1)\sum_{i=\omega-t}^{\infty}C(i,1)\pi_i.
\]
Plugging in the definitions yields
\[
\Sigma_t - \Delta_t'\Pi_t = p_tP^t_{1,1+t}(1)\left(\sum_{i=\omega-t}^{\omega-1}C(i,1)\pi_i+\Sigma\right).
\]
Summing the above equation over $t$ from $1$ to $t_{max}$ yields
\begin{multline*}
\sum_{t=1}^{t_{max}}\bigg(\Sigma_t - \Delta_t'\Pi_t\bigg) =\\
\sum_{t=1}^{t_{max}}\left[p_tP^t_{1,1+t}(1)\left(\sum_{i=\omega-t}^{\omega-1}C(i,1)\pi_i+\Sigma\right)\right].
\end{multline*}
Rearranging the above equation yields
\begin{equation}\label{eq-SigmaRaw}
\begin{split}
\Sigma - \sum_{t=1}^{t_{max}}\Delta_t'\Pi_t = & \sum_{t=1}^{t_{max}}\left[p_tP^t_{1,1+t}(1)\left(\sum_{i=\omega-t}^{\omega-1}C(i,1)\pi_i\right)\right] +\\
& \sum_{t=1}^{t_{max}}\bigg(p_tP^t_{1,1+t}(1)\bigg)\Sigma.
\end{split}
\end{equation}
Hence, the closed-form expression of $\Sigma$ is
\[
\Sigma = \ddfrac{\sum_{t=1}^{t_{max}}\left[p_tP^t_{1,1+t}(1)\left(\sum_{i=\omega-t}^{\omega-1}C(i,1)\pi_i\right) + \Delta_t'\Pi_t\right]}{1-\sum_{t=1}^{t_{max}}\bigg(p_tP^t_{1,1+t}(1)\bigg)}.
\]
In the following, we calculate $\Pi_t$. Combining the definition of $\Pi_t$ with \eqref{eq-0util41}, we have
\begin{align*}
\Pi_t\triangleq \sum_{i=\omega}^{\infty}\pi_{i,t} = & \sum_{i=\omega}^{\infty}\bigg(p_{t}P^{t}_{i-t,i}(1)\pi_{i-t}\bigg)\\
= & \sum_{i=\omega-t}^{\infty}\bigg(p_{t}P^{t}_{i,i+t}(1)\pi_{i}\bigg).
\end{align*}
Since $P^{t}_{\Delta,\Delta+t}(1)$ is independent of $\Delta$ when $\Delta>0$, we have
\[
\Pi_t = p_{t}P^{t}_{1,1+t}(1)\left(\sum_{i=\omega-t}^{\omega-1}\pi_{i} + \Pi\right).
\]
Combining together, we recover the results presented in the theorem.

\section{Proof of Theorem \ref{prop-ThresholdPerformance}}\label{pf-ThresholdPerformance}
We follow similar steps as in the proof of Theorem \ref{prop-Performance}. First, $\pi_\Delta$ satisfies
\begin{align*}
\pi_{\Delta} = & \sum_{i=\Delta-t_{max}}^{\Delta-1}P_{i,\Delta}(1)\pi_i \\
= & \sum_{i = 1}^{t_{max}}P_{\Delta-t_{max}+i-1,\Delta}(1)\pi_{\Delta-t_{max}+i-1},\quad \Delta\geq \omega.
\end{align*}
We recall from Lemma \ref{lem-Case2TransProb}, $P_{\Delta,\Delta'}(1) = p_{t'}P^{t'}_{\Delta,\Delta'}(1) + p_{t^+}P^{t^+}_{\Delta,\Delta'}(1)$ where $t' = \Delta'-\Delta$ when $\Delta'\geq\omega$. Then, we have \eqref{eq-EquivalentEq8} holds.
\begin{figure*}[!t]
\normalsize
\begin{equation}\label{eq-EquivalentEq8}
\pi_\Delta = \sum_{i=1}^{t_{max}}\bigg(p_{t_{max}+1-i}P^{t_{max}+1-i}_{\Delta-t_{max}+i-1,\Delta}(1) + p_{t^+}P^{t^+}_{\Delta-t_{max}+i-1,\Delta}(1)\bigg)\pi_{\Delta-t_{max}-1+i},\quad \Delta\geq\omega.
\end{equation}
\hrulefill
\vspace*{4pt}
\end{figure*}
Renaming the variables yields
\[
\begin{split}
\pi_\Delta = & \sum_{t=1}^{t_{max}}\bigg(p_tP^t_{\Delta-t,\Delta}(1) + p_{t^+}P^{t^+}_{\Delta-t,\Delta}(1)\bigg)\pi_{\Delta-t}\\
= & \sum_{t=1}^{t_{max}}\Upsilon(\Delta,t)\pi_{\Delta-t},\quad \Delta\geq\omega,
\end{split}
\]
where $\Upsilon(\Delta,t)\triangleq p_tP^t_{\Delta-t,\Delta}(1) + p_{t^+}P^{t^+}_{\Delta-t,\Delta}(1)$. We notice that $\Upsilon(\Delta,t)$ is independent of $\Delta$ when $\Delta \geq \omega$. To proceed, we define, for each $1\leq t\leq t_{max}$,
\[
\pi_{\Delta,t} \triangleq \Upsilon(\Delta,t)\pi_{\Delta-t},\quad\Delta\geq\omega.
\]
Note that $\sum_{t=1}^{t_{max}}\pi_{\Delta,t} = \pi_{\Delta}$. Then, for a given $1\leq t\leq t_{max}$, we have
\begin{equation}\label{eq-ZeroWaitPolicy4}
\sum_{i=\omega}^{\infty}C(i-t,1)\pi_{i,t} = \sum_{i=\omega}^{\infty}C(i-t,1)\Upsilon(i,t)\pi_{i-t}.
\end{equation}
We define $\Delta_t' \triangleq C(\Delta,1) - C(\Delta-t,1)$ where $\Delta>t$. Then, according to \eqref{eq-CompactCostRan2}, we have
\begin{align*}
\Delta_t' = & \sum_{i=1}^{t_{max}}p_i\bigg(C^i(\Delta,1) - C^i(\Delta-t,1)\bigg) + \\
& p_{t^+}\bigg(C^{t_{max}}(\Delta,1) - C^{t_{max}}(\Delta-t,1)\bigg).
\end{align*}
By Lemma \ref{lem-CompactCost}, we have \eqref{eq-EquivalentEq6} and \eqref{eq-EquivalentEq7}.
\begin{figure*}[!t]
\normalsize
\begin{equation}\label{eq-EquivalentEq6}
C^{i}(\Delta-t,1) = \Delta - t + \sum_{h=1}^{i-1}\left(\sum_{k=1}^{h-1} k(1-p^{(h-k)})p(1-p)^{k-1} + (\Delta-t+h)(1-p)^h\right).
\end{equation}
\begin{equation}\label{eq-EquivalentEq7}
C^{i}(\Delta,1) = \Delta+ \sum_{h=1}^{i-1}\left(\sum_{k=1}^{h-1} k(1-p^{(h-k)})p(1-p)^{k-1} + (\Delta+h)(1-p)^h\right).
\end{equation}
\hrulefill
\vspace*{4pt}
\end{figure*}
Subtracting the two equations yields
\begin{align*}
C^{i}(\Delta,1) - C^{i}(\Delta-t,1) = & t + \sum_{h=1}^{i-1}\bigg(t(1-p)^h\bigg)\\
= & \frac{t-t(1-p)^i}{p}.
\end{align*}
Then, for each $1\leq t\leq t_{max}$, we have
\[
\Delta_t' = \sum_{i=1}^{t_{max}}p_i\left(\frac{t-t(1-p)^i}{p}\right) + p_{t^+}\left(\frac{t-t(1-p)^{t_{max}}}{p}\right).
\]
We notice that $\Delta_t' = C(\Delta,1) - C(\Delta-t,1)$ is independent of $\Delta$ when $\Delta>t$. Hence, equation \eqref{eq-ZeroWaitPolicy4} can be written as
\[
\sum_{i=\omega}^{\infty}\bigg(C(i,1) - \Delta_t'\bigg)\pi_{i,t} = \sum_{i=\omega-t}^{\infty}C(i,1)\Upsilon(i+t,t)\pi_i.
\]
Then, we define $\Pi_t \triangleq \sum_{i=\omega}^{\infty}\pi_{i,t}$ and $\Sigma_t \triangleq \sum_{i=\omega}^{\infty}C(i,1)\pi_{i,t}$. We recall that $\Upsilon(\Delta,t)$ is independent of $\Delta$ when $\Delta\geq \omega$. Hence, plugging in the definitions yields
\[
\Sigma_t - \Delta_t'\Pi_t = \sum_{i=\omega-t}^{\omega-1}\Upsilon(i+t,t)C(i,1)\pi_i+\Upsilon(\omega+t,t)\Sigma.
\]
Summing the above equation over $t$ from $1$ to $t_{max}$ yields
\begin{multline*}
\sum_{t=1}^{t_{max}}\bigg(\Sigma_t - \Delta_t'\Pi_t\bigg) =\\
\sum_{t=1}^{t_{max}}\left(\sum_{i=\omega-t}^{\omega-1}\Upsilon(i+t,t)C(i,1)\pi_i+\Upsilon(\omega+t,t)\Sigma\right).
\end{multline*}
Rearranging the above equation yields
\begin{align*}
\Sigma - \sum_{t=1}^{t_{max}}\Delta_t'\Pi_t = & \sum_{t=1}^{t_{max}}\left(\sum_{i=\omega-t}^{\omega-1}\Upsilon(i+t,t)C(i,1)\pi_i\right) +\\
& \sum_{t=1}^{t_{max}}\Upsilon(\omega+t,t)\Sigma.
\end{align*}
Then, the closed-form expression of $\Sigma$ is
\[
\Sigma = \ddfrac{\sum_{t=1}^{t_{max}}\left[\left(\sum_{i=\omega-t}^{\omega-1}\Upsilon(i+t,t)C(i,1)\pi_i\right) + \Delta_t'\Pi_t\right]}{1-\sum_{t=1}^{t_{max}}\Upsilon(\omega+t,t)}.
\]
In the following, we calculate $\Pi_t$. To this end, we have
\[
\Pi_t\triangleq \sum_{i=\omega}^{\infty}\pi_{i,t} = \sum_{i=\omega}^{\infty}\Upsilon(i,t)\pi_{i-t} = \sum_{i=\omega-t}^{\infty}\Upsilon(i+t,t)\pi_i.
\]
Since $\Upsilon(\Delta,t)$ is independent of $\Delta$ if $\Delta \geq \omega$, we have
\[
\Pi_t = \sum_{i=\omega-t}^{\omega-1}\Upsilon(i+t,t)\pi_i + \Upsilon(\omega+t,t)\Pi,\quad 1\leq t\leq t_{max}.
\]
Combining together, we recover the results presented in the theorem.

\end{document}